\documentclass[twoside]{article}

\usepackage[accepted]{aistats2022}

\usepackage[utf8]{inputenc} 
\usepackage[T1]{fontenc}
\usepackage{url}
\usepackage{ifthen}
\usepackage{cite}
\usepackage{graphicx,colordvi,psfrag}
\usepackage{amsmath,amssymb}
\usepackage{epstopdf}
\usepackage{epsfig}
\usepackage{calc,pstricks, pgf, xcolor}
\usepackage{bbding}
\usepackage{bbm}
\usepackage{dsfont}
\usepackage{centernot}
\usepackage{bm}
\usepackage{delarray, multirow}
\usepackage{enumitem}
\usepackage[colorlinks=true]{hyperref}
\usepackage{authblk}
\usepackage{mathtools}
\usepackage{setspace}
\usepackage[ruled,vlined,linesnumbered]{algorithm2e}
\usepackage{faktor}
\usepackage{caption}
\usepackage{subcaption}

\newcommand{\Ind}{\mathds{1}}

\newcommand{\by}{\mathbf{y}}
\newcommand{\bY}{\mathbf{Y}}
\newcommand{\bx}{\mathbf{x}}
\newcommand{\bX}{\mathbf{X}}

\newcommand{\bZ}{\mathbf{Z}}
\newcommand{\bh}{\mathbf{h}}
\newcommand{\bH}{\mathbf{H}}
\newcommand{\bA}{\mathbf{A}}
\newcommand{\bN}{\mathbf{N}}

\newcommand{\ba}{\mathbf{a}}

\newcommand{\bu}{\mathbf{u}}

\newcommand{\bv}{\mathbf{v}}

\newcommand{\bt}{\mathbf{t}}

\newcommand{\be}{\mathbf{e}}

\newcommand{\bI}{\mathbf{I}}

\newcommand{\bp}{\mathbf{p}}
\newcommand{\bq}{\mathbf{q}}

\newcommand{\bS}{\mathbf{S}}

\newcommand{\ZZ}{\mathbb{Z}}
\newcommand{\RR}{\mathbb{R}}

\newcommand{\Var}{\mathrm{Var}}
\newcommand{\Cov}{\mathrm{Cov}}

\newcommand{\rank}{\mathop{\mathrm{rank}}}

\newcommand{\Vol}{\mathrm{Vol}}

\newcommand{\Expt}{\mathbb{E}}
\newcommand{\eps}{\varepsilon}
\newcommand{\Unif}{\mathrm{Unif}}
\newcommand{\poly}{\mathrm{poly}}

\newcommand{\m}{\mathcal}

\newcommand{\T}{\top}

\newcommand{\0}{\bm{0}}

\DeclareMathOperator*{\argmin}{\arg\!\min}

\newtheorem{theorem}{Theorem}

\newtheorem{proposition}{Proposition}

\newtheorem{definition}{Definition}
\newtheorem{lemma}{Lemma}

\newenvironment{proof}[1][Proof]{\noindent\textbf{#1.} }{\ \rule{0.5em}{0.5em}}

\newcommand{\bSigma}{{\bf{\Sigma}}}
\newcommand{\bSigmaTruc}{\bm{\Sigma}_{\mathrm{Ball}}}
\newcommand{\hbSigma}{\widehat{\bSigma}}
\newcommand{\hbSigmaX}{\hbSigma_{\bX}}

\newcommand{\SphereK}{\mathbb{S}^{k-1}}
\newcommand{\VolK}{\m{V}_k}

\newcommand{\Id}{\bm{I}}
\newcommand{\bXTruc}{\bX_{\mathrm{Ball}}}
\newcommand{\Spike}{\nu}
\newcommand{\SpikeVec}{\bu}
\newcommand{\SpikeVecEst}{\widehat{\bu}}

\newcommand{\PBall}{p_{\mathrm{Ball}}}
\newcommand{\Kpicked}{\m{K}}
\newcommand{\Kball}{\m{K}_{\mathrm{Ball}}}
\newcommand{\Kgood}{\m{K}_{\mathrm{Good}}}
\newcommand{\Kbad}{\m{K}_{\mathrm{Bad}}}

\newcommand{\epsEst}{\eps_{\mathrm{CovEst}}}
\newcommand{\epsPick}{\eps_{\mathrm{Pick}}}

\newcommand{\erf}{\mathrm{erf}}

\newcommand{\Cube}{\m{Q}_{\Delta}}

\newcommand{\ConstDelta}{C_{*}}

\begin{document}

\twocolumn[

\aistatstitle{Spiked Covariance Estimation from Modulo-Reduced Measurements}

\aistatsauthor{ 
	Elad Romanov 
	\And 
	Or Ordentlich 
	}

\aistatsaddress{
	Hebrew University of Jerusalem
	\And 
	Hebrew University of Jerusalem
	} ]

\begin{abstract}
	Consider the {rank-1} spiked model: ${\bX=\sqrt{\Spike}\xi \SpikeVec + \bZ}$, where $\Spike$ is the spike intensity, $\SpikeVec\in\SphereK$ is an unknown direction and ${\xi\sim \m{N}(0,1),\bZ\sim \m{N}(\0,\Id)}$. 
	Motivated by recent advances in analog-to-digital conversion, we study the problem of recovering $\bu\in \SphereK$ from $n$ i.i.d. \emph{modulo-reduced} measurements $\bY=[\bX]\mod \Delta$, focusing on the high-dimensional regime ($k\gg 1$). We {develop and analyze} an algorithm that, for most directions $\bu$ and $\Spike=\poly(k)$, {estimates} $\bu$ {to high accuracy} using $n=\poly(k)$ measurements, provided that $\Delta\gtrsim \sqrt{\log k}$.  Up to constants, our algorithm {accurately estimates $\bu$} at the smallest possible $\Delta$ that allows (in an information-theoretic sense) to recover $\bX$ from $\bY$.
	 {A key step in our analysis involves estimating the probability that a line segment of length $\approx\sqrt{\nu}$ in a random direction $\bu$ passes near a point in the lattice $\Delta \mathbb{Z}^k$.} 
	 {Numerical experiments show that the developed algorithm performs well even in a non-asymptotic setting.}
\end{abstract}

\section{Introduction}\label{sec:intro}

We consider the problem of estimating a spiked covariance matrix from Gaussian modulo-folded measurements. Let $\SpikeVec\in\SphereK$ be an unknown direction, and $\Spike>0$ be the signal-to-noise (SNR) ratio. Consider the spiked covariance matrix
\begin{equation}\label{eq:Sigma}
	\bSigma = \Spike \SpikeVec\SpikeVec^\T + \Id \,,
\end{equation}
and denote $\bX\sim \m{N}(\0,\bSigma)$. Equivalently, one may write
\begin{equation}\label{eq:X}
	\bX = \sqrt{\Spike}\xi \bu + \bZ\,,
\end{equation}
where $\xi,\bZ$ {have} $\m{N}(0,1)$ entries; the one dimensional component $\sqrt{\Spike}\xi\bu$ is often thought of as the ``signal'', whereas $\bZ$ is thought of as {``noise''}. In this paper, we consider the problem of estimating $\bu$ from $n$ independent and modulo-reduced measurements of $\bX$. Let $\Delta>0$ be the \emph{dynamic range}, and for ${X\in \RR}$, denote the modulo operation by
\begin{eqnarray}
	Y = [X]\bmod \Delta \in \left[-\frac12 \Delta,\frac12\Delta\right)\,,
\end{eqnarray}
so that $Y$ is the unique number in the half-open interval such that $X-Y\in \Delta \ZZ$. For a vector ${\bX\in\RR^k}$, $\bY = [\bX]\bmod \Delta$ is defined by modulo-reducing each coordinate separately. 
In the setup we consider, one is given $n$ independent copies of $\bY$, denoted by
$\by_1,\ldots,\by_n$, and wishes to estimate the unknown direction $\SpikeVec\in\SphereK$. Throughout, we denote by $\bx_1,\ldots,\bx_n$ independent copies of $\bX$, such that $\by_i=[\bx_i]\bmod \Delta$. 
See Figure~\ref{fig:Figure1} for a graphical illustration, in $k=2$ dimensions. 

\begin{figure}[h]
	\centering
	\includegraphics[width=0.5\textwidth,height=0.2\textheight,keepaspectratio]{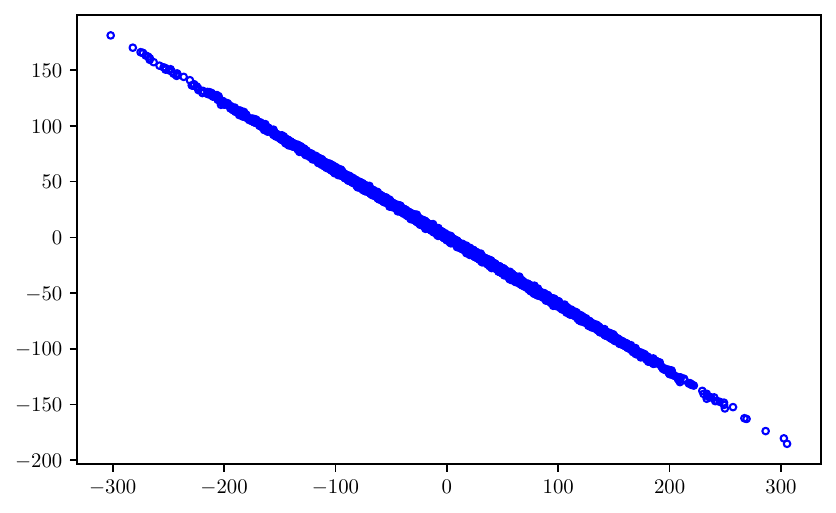}
	\includegraphics[width=0.5\textwidth,height=0.2\textheight,keepaspectratio]{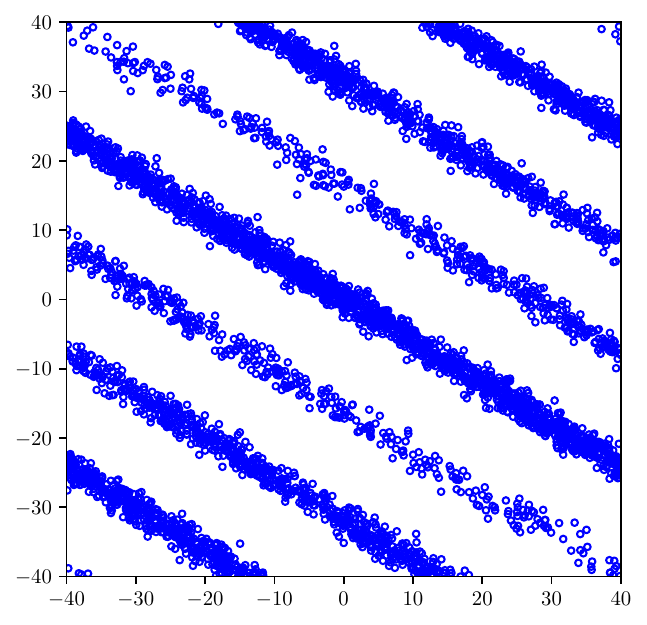}
	\caption{A typical problem instance in $k=2$ dimensions. Top: A point cloud, corresponding to $n=5000$ i.i.d. samples from $\bX\sim \m{N}(\0,\Spike\SpikeVec\SpikeVec^\T + \Id)$, for $\Spike=10^4$ and some fixed $\SpikeVec\in\SphereK$. Bottom: The modulo-reduced point cloud, with $\Delta=80$.}
	\label{fig:Figure1}
\end{figure}

\paragraph{Motivation.}
 Our motivation for considering this problem is driven by recent developments in signal processing. Analog-to-digital converters (ADCs), devices that convert analog (continuous) signals into digital (discrete, e.g. bits) signals, are an essential component in virtually all modern communication devices. 
 From a mathematical perspective, ADCs are set out to solve {\it essentially} the following problem: Given a vector-valued random variable $\bX\in \RR^k$, find a quantizer (binning scheme), $\phi:\RR^k\to \RR^k$, with a finite range, $|\mathrm{range}(\phi)|\le 2^B$ ($B$ being the allowable representation length in bits), so as to minimize the quantization error: $\Expt \|\bX-\phi(\bX)\|^2$. 
 Quantization is an extensively studied problem, and its fundamentals limits (in information theory: ``rate-distortion theory''), under setups of varying generality, are largely understood; see, e.g., \cite{cover2012elements,gersho2012vector}. 
 
Optimal vector quantizers (that achieve the fundamental limits) typically involve rather complicated constructions, that depend intricately on the \emph{exact} statistics of $\bX$. Since ADCs are implemented in mixed analog-digital circuits, sophisticated vector quantizers are prohibitive and the design is often 
restricted to architectures of a scalar product form: ${\phi(\bX)=(\phi_1(X_1),\ldots,\phi_k(X_k))}$ where ${\phi_1,\ldots,\phi_k:\RR\to\RR}$. Perhaps the simplest -- and most popular -- architecture in practice is a uniform scalar quantizer, determined by its dynamic range $\Delta>0$ and bit-rate $b=B/k$.  The quantizer divides  the interval $[-\Delta/2,\Delta/2]$ into $2^{b}$ intervals of equal size $\Delta 2^{-b}$, so that the scalar quantizer $\phi_1=\ldots=\phi_k=\phi$ maps $X\in \RR$ to its closest interval center. Note that for this scheme 
to attain quantization error that vanishes as the quantization rate $b$ increases, 
the dynamic range must be $\Delta \gtrsim \max_{1\le i \le k}\sqrt{\Var(X_i)}$.\footnote{If $X$ falls inside $[-\Delta/2,\Delta/2]$, then $|X-\phi(X)|\le \Delta 2^{-b}$. For the expected error to be of the same order, the probability of a {\it saturation}, namely $|X|>\Delta/2$, has to be small.}

Although simple to implement, the uniform quantizer can be pronouncedly sub-optimal for vector-valued signals, as it cannot leverage the cross-coordinate correlations that often occur in real-world applications. An important use-case in digital communications is Massive MIMO, where typically the number of users is much smaller than the number of receive antennas; this  results in signals $\bX$ that have strong cross-entry correlations.\footnote{A multi-user MIMO channel is modeled by ${\bX=\bH \bS + \bZ}$, where $k$ the number of receive antennas, $m$ is the number of transmitting users, each equipped with a single antenna, and $\bH=[\bh_1|\cdots|\bh_m]\in \RR^{k\times m}$ is the channel matrix ($\bh_i\in\RR^{k}$ represents the channel gains from transmitter $i$ to the receiver). The  vector  $\bS=[S_1,\ldots,S_m]\in \RR^{m}$ represents the transmissions of the $m$ users, and $\bZ\in \RR^k$ is white noise. The ``signal'' part, $\bH\bS$, lives inside an $m$-dimensional subspace in $\RR^k$, where typically $m \ll k$. One example, among many, for the extreme case of rank $m=1$ MIMO (corresponding to model (\ref{eq:X}) exactly), is in low-earth orbit (LEO) communications, where a phased array receiver 
is used to track a rapidly moving satellite.  } 
Thus, a quantization scheme that can exploit these statistical inter-dependencies, while retaining the simplicity of the uniform quantizer, is highly desired. 

A recently proposed architecture, ``modulo-ADCs'' \cite{ordentlich2018modulo}, attempts to address these issues. Their idea is rather simple: do not \emph{truncate} $\bX$ onto $[-\Delta/2,\Delta/2]^k$, as the uniform quantizer does;  
instead, apply \textbf{modulo-reduction} $\bY=[\bX]\bmod \Delta \in [-\Delta/2,\Delta/2]^k$ and then quantize as before. 
For this idea to work, one clearly need some means of ``unwrapping'' $\bX$ from $\bY$ (with high probability).
When the coordinates of $\bX$ are independent and unimodal, with the mode at $0$ (for example, a centered Gaussian), it is easy to see that the best estimator for $\bX$ from $\bY$ (in the sense of error probability) is just $\widehat{\bX}=\bY$. Thus, a coordinate $X_i$ cannot be recovered once it saturates the ADC dynamic range, $|X_i|\ge \Delta/2$; so
to consistently undo the modulo, one needs $\Delta \gtrsim \max_{1\le i \le k}\sqrt{\Var(X_i)}$, and the scheme has no advantage over the standard uniform quantizer. It turns out, however, that when $\bX$ has strong correlations, it is often possible to consistently unwrap at \emph{substantially} smaller values of $\Delta$; see next section.

\subsection{Related work}

We start with very brief background on the spiked model, Eq. (\ref{eq:X}). In the high-dimensional statistics ($k\approx n$, $k,n\to\infty$) literature, the spiked model was popularized by \cite{johnstone2001distribution}, who studied the largest eigenvalue of the sample covariance matrix $\hbSigma=\frac1n \sum_{i=1}^n \bx_i\bx_i^\T$. Subsequent advances in random matrix theory \cite{baik2005phase,paul2007asymptotics} {characterized} the behavior of PCA (namely, the relation between the principal components of $\bSigma$ and its empirical counterpart $\hbSigma$) rather precisely. Since then, a vast literature on the spiked model (and variations thereof) has emerged -- which we make no pretense to survey here; as an entry point, geared towards statisticians, see \cite[Chapter 8]{wainwright2019high}. We cite the following minimax lower bound for the spike estimation problem (without modulo-reduction) \cite[Example 15.19]{wainwright2019high}:
\begin{equation}
	\label{eq:minimax-rate}
	\left\| \SpikeVec-\SpikeVecEst \right\| \gtrsim  \left(\frac{\sqrt{1+\Spike}}{\Spike} \sqrt{\frac{k}{n}}\right) \wedge 1 \,.
\end{equation}
In the regime $k=O(n)$, 
this rate is attained, up to prefactors, 
by PCA ($\SpikeVecEst$ {taken to be} the largest eigenvector of $\hbSigma$), see \cite[Corollary 8.7]{wainwright2019high}.

Moving on,
there has recently been a great deal of activity in the signal processing community around recovery from modulo-reduced measurements~\cite{bhandari2017unlimited,ordentlich2018modulo,bkr18isit,gbk19,romanov2019above,bk19,bhandari2020unlimited,bkp21,weiss2021blind}.
Most relevant to this paper is a line of works dealing with recovery from modulo-reduced measurements, and motivated by the modulo-ADC architecture described before \cite{oe17,ordentlich2018modulo}. The setting is this: the source is \emph{Gaussian} $\bX\sim \m{N}(\0,\bSigma)$, with $\bSigma$ some covariance matrix (not necessarily spiked); one observes modulo-reduced measurements $\bY=[\bX]\mod \Delta$, and wishes to recover $\bX$ itself (with high probability).
How large should $\Delta$ be so that consistent recovery is possible, in an information-theoretic sense? 
When it \emph{is} possible, how could one do so \emph{practically}? (Assuming $\bSigma$ is known? And when it is \emph{not}?) 
The answers, it turns out, depend rather intricately on the \emph{diophantine} properties of the matrix $\bSigma$. 

Let us start with the fundamental limits. A simple observation 
is that when $\Delta\gtrsim \sqrt{\log k} \cdot \max_{1\le \ell \le k}\sqrt{\bSigma_{\ell,\ell}}$ ($\bSigma_{\ell,\ell}$ being the variance of the $\ell$-th coordinate), one has $\bX=\bY$ with high probability, so that consistent recovery is straightforward. It is easy to see that when $\bX$ is white, in other words $\bSigma\propto \Id$, this requirement is in fact tight. For general $\bSigma$, one may readily show that the maximum a posteriori probability (MAP) estimator for $\bX$ given $\bY$ is 
\begin{equation}\label{eq:MAP}
	\widehat{\bX}_{\mathrm{MAP}}(\bY=\by) = \argmin_{\bx\,:\,\bx-\by\in\Delta \ZZ^k} \bx^\T \bSigma^{-1}\bx\,,
\end{equation}
that is, one needs to minimize a quadratic form over the coset of $\by$. Searching over the coset directly {(and consequently, computing $\widehat{\bX}_{\mathrm{MAP}}$ exactly)} is not computationally tractable, in all but the simplest cases; nonetheless, since $\widehat{\bX}_{\mathrm{MAP}}$ is optimal in the sense of error probability, its performance characterizes the information-theoretic limits of the problem. The latter has a rather elegant geometric interpretation. Let $\m{L}=\Delta\bSigma^{-1/2}\ZZ^k\subseteq \RR^k$ be the lattice generated by $\Delta\bSigma^{-1/2}\in \RR^{k\times k}$, and let $V_0\subseteq \RR^k$ be the Voronoi cell of $\0\in \m{L}$. Then \cite{romanov2021blind} the success probabililty of (\ref{eq:MAP}) is the Gaussian measure of $V_0$: 
\[
p_{\mathrm{MAP}}=\Pr\left( \widehat{\bX}_{\mathrm{MAP}}  =\bX\right) = \Pr_{\bZ\sim \m{N}(\0,\Id)} (\bZ\in V_0) \,.
\]
As a corollary, it is not hard to show that\footnote{$|\bSigma|$ denotes the determinant of $\bSigma$.} $\Delta \lesssim |\bSigma|^{1/2k}$ implies that $p_{\mathrm{MAP}}=o(1)$; see also Proposition~\ref{prop:lower-bound}. When the lattice $\m{L}$ is a \emph{uniformly random lattice}, sampled, up to normalization, from the Haar measure over $\faktor{\mathrm{SL}_k(\RR)}{\mathrm{SL}_k(\ZZ)}$ (also called Haar-Siegel measure), one can show that $V_0$ is with high probability ``sufficiently ball-like'', so that $\Delta\gtrsim |\bSigma|^{1/2k}$ is also a sufficient condition. Random lattices have played a prominent role in the lattice coding literature \cite{zamir2014lattice}, which is closely related to the present line of work. An important point is that ``natural'' random matrix ensembles, such as the spiked ensemble (\ref{eq:Sigma}) with $\bu\sim \Unif(\SphereK)$, \textbf{do not} 
correspond to the Siegel-Haar measure on the space of lattices.
In \cite{domanovitz2017outage} the authors demonstrate that certain orthogonally-invariant ensembles, that arise in channel coding theory, indeed allow for consistent recovery with $\Delta$ not much larger than $|\bSigma|^{1/2k}$. In particular, for the spiked ensemble (\ref{eq:Sigma}), they show that with high probability over $\SpikeVec\sim \Unif(\SphereK)$, the error probability is small whenever $\Delta\gtrsim C(k)|\bSigma|^{1/2k}=C(k)(1+\Spike)^{1/2k}$, where $C(k)$ grows exponentially fast in $k$, but does not depend on the SNR $\Spike$.

As for practical recovery algorithms, let us start by assuming $\bSigma$ is \textbf{known}. {As already mentioned, } computing the MAP estimator directly is intractable; instead \cite{oe17} proposed to use a sub-optimal estimator, the so-called {\it Integer Forcing} (IF) decoder. The idea is to find an invertible integer matrix $\bA=\left[ \ba_1,\ldots,\ba_k \right]$ so as to minimize the maximal variance:
{
\begin{equation}\label{eq:IF}
	\bA_{\mathrm{IF}}=\argmin_{\substack{\bA\in\ZZ^{k\times k},\\\rank(\bA)=k}} \max_{1\le \ell \le k} {\ba_\ell ^\T \bSigma \ba_\ell}\,.
\end{equation}
}
The {corresponding} maximal deviation, $m_k(\bSigma)=\max_{1\le \ell \le k} \sqrt{\ba_\ell ^\T \bSigma \ba_\ell}$, is called the $k$th successive minimum of the lattice $\bSigma^{1/2}\ZZ^k$. 
Since $[\bA\bY] \bmod \Delta = [\bA \bX] \bmod \Delta$, one can reliably recover $\bX$ from $\bY$, using $\widehat{\bX}_{\mathrm{IF}} = \bA_{\mathrm{IF}}^{-1}\left([\bA_{\mathrm{IF}}\bY]\mod \Delta \right)$, whenever $\Delta\gtrsim m_k(\bSigma)\sqrt{\log k}$.
 {Of course,}
to compute the IF decoder, Eq. (\ref{eq:IF}), one clearly needs to know $\bSigma$.\footnote{An important caveat is that Eq.(\ref{eq:IF}) is actually a computationally hard problem, and may only be solved exactly for very small $k$. In practice, one usually solves this {\it approximately}, using a lattice reduction algorithm, like the Lenstra-Lenstra-Lov{\'a}sz (LLL) algorithm \cite{lenstra1982factoring}. Observe that if one a priori restricts the minimization to $\bA\in\mathrm{SL}_k(\ZZ)$, then Eq. (\ref{eq:IF}) is equivalent to finding a shortest basis for the lattice $\bSigma^{1/2}\ZZ^k$.} In some applications, for example in wireless communications (where $\bSigma$ depends on the channel matrix, which rapidly changes over time) \cite{tse2005fundamentals}, this is \emph{not} a reasonable assumption. In \cite{romanov2021blind}, the authors propose a \textbf{blind} unwrapping algorithm, that does not know $\bSigma$ beforehand, in a setting where one needs to simultaneously unwrap many i.i.d. signals $\by_1,\ldots,\by_n$. A natural step towards that end is to estimate $\bSigma$ from the (modulo-reduced) data, from which the integer-forcing decoder (\ref{eq:IF}) could be computed. Alas, directly computing the maximum likelihood estimator (MLE) of $\bSigma$ from modulo-reduced measurements is not computationally feasible.
 Instead, they propose an algorithm which iteratively alternates between 1) A covariance estimation step, where a certain ``proxy'' of $\bSigma$ is estimated; 2) An integer forcing decoder, computed from that proxy; the idea is to gradually ``whiten'' the entire dataset, effectively computing the IF decoder ``in small steps''. They prove a result of the following form:
%
 when the error of the informed IF decoder (Eq. (\ref{eq:IF})) is {\it small enough}, then the error of the adaptive algorithm is {\it essentially comparable} to it, up to dimension-dependent prefactors.
However, their algorithm is only suited to rather modest $k$, as seen both in the analysis (the prefactors are exponential in $k$) and the numerical experiments. The problem lies with their covariance estimation procedure, whose performance breaks down rapidly as the dimension increases.
This is the starting point for the present paper.

\paragraph{Our contributions.} We propose a computationally tractable algorithm to estimate the spike ${\SpikeVec\in \SphereK}$ from modulo-reduced measurements, under the spiked covariance model Eq. (\ref{eq:Sigma}) and in high dimension $k$. We show that for \emph{most} directions $\SpikeVec$ (formally: with high probability over ${\bu\sim \Unif(\SphereK)}$), estimation is possible with $n=\poly(k)$ samples, under essentially the smallest $\Delta$ (up to constants) that allows for consistent unwrapping. 
Thus, in this setting, we \emph{provably} overcome the curse of dimensionality suffered by the algorithm of \cite{romanov2021blind}.
While we do not directly tackle the unwrapping problem, note that in applications where the SNR $\Spike$ is approximately known, the algorithm readily yields an estimate for $\bSigma$, from whence one could compute the IF decoder (\ref{eq:IF}).  
Our numerical experiments below show that this method attains an unwrapping error probability that is not far from that of the informed IF decoder.

\paragraph{Notation.} For sequences $a_k,b_k$ we use the following standard notation: $a_k\vee b_k = \max\{a_k,b_k\}$, $a_k\wedge b_k=\min\{a_k,b_k\}$. By $a_k\lesssim b_k$ we mean that $a_k\le Cb_k$ for some \emph{universal} constant $C>0$; we write $a_k\approx b_k$ whenever both $a_k\lesssim b_k$ and $b_k\lesssim a_k$. We also use \mbox{big-O} notation; if $M$ is a parameter, we use $a_k=O_M(b_k)$ to signify that the constants might depend on $M$. For a vector $\bv\in \RR^k$, $\|\bv\|$ denotes its $\ell_2$ (Euclidean) norm.

\section{Proposed method}
 
 
 \paragraph{An observation.}
Our algorithm is based on the following observation: the eigen-structure of the covariance matrix of $\bX$ is preserved when truncated onto a ball. 
Set $R>0$ a  truncation radius. For ${\bX\sim \m{N}(\0,\bSigma)}$, let $\bXTruc$ be its spherically-truncated version:
 for ${\m{S}\subseteq \RR^k}$,
 \begin{align*}
\Pr(\bXTruc\in \m{S}) 
&= \Pr(\bX\in\m{S}|\bX\in\m{B}(\0,R)) \\
&=\frac{\Pr(\bX\in \m{S}\cap \m{B}(\0,R))}{\Pr(\bX\in \m{B}(\0,R))} \,.
 \end{align*}
Observe that $\Expt[\bXTruc]=\0$, since $\bXTruc$ is symmetric. Denote the covariance by $\bSigmaTruc=\Expt[\bXTruc\bXTruc^\T]$.
\begin{proposition}\label{obs:trunc-eigs}
	Let $\bSigma$ be any covariance matrix, with (orthonormal) eigenvectors ${\bu_1,\ldots,\bu_k\in\SphereK}$ and corresponding eigenvalues ${\lambda_1\ge \ldots \ge \lambda_k \ge 0}$. Then 
	\begin{enumerate}
		\item \label{item:obs1} The basis $\bu_1,\ldots,\bu_k$ diagonalizes $\bSigmaTruc$. Denote $g_1,\ldots,g_k\overset{i.i.d.}{\sim} \m{N}(0,1)$; the respective eigenvalues are:
		\begin{equation}\label{eq:trunc-eigs}
			\mu_i = \Expt \left[ \lambda_i g_i^2 \,\Big| \, \sum_{i=1}^k \lambda_i g_i^2 \le R^2\right]\,.
		\end{equation}
		\item \label{item:obs2} The ordering is preserved: ${\mu_1\ge \ldots \ge \mu_k}$.
	\end{enumerate}
\end{proposition}
Proposition~\ref{obs:trunc-eigs} is not new by any means. It has appeared before in \cite{palombi2012numerical}, which considered covariance estimation from spherically-truncated Gaussian measurements (see also discussion later in this section). The proof of Item~\ref{item:obs1} is rather trivial; for completeness, we provide a short proof  (Appendix~\ref{sec:proof-obs:trunc-eigs}). Item \ref{item:obs2} is considerably less so; we refer to \cite[Proposition 3.3]{palombi2012numerical} for the details.


\paragraph{The algorithm.}
We draw inspiration from Figure~\ref{fig:Figure1}. For ``most'' directions $\SpikeVec\in\SphereK$, 
the points $\{\by_i\}_{i=1}^n$ are arranged, essentially, in parallel and separated stripes. The central stripe (that crosses the origin) consists of points that have not undergone folding, $\by_i=\bx_i$. Picking only points 
inside a \emph{small enough} ball ${\m{Y}_{\mathrm{Ball}}=\{\by_i\}_{i=1}^n \cap \m{B}(\0,R)}$, we therefore obtain, approximately, an i.i.d. sample from $\bXTruc$. By Proposition~\ref{obs:trunc-eigs}, the leading eigenvector of $\bSigmaTruc$ is $\SpikeVec$, and therefore PCA with $\m{Y}_{\mathrm{Ball}}$ should yield a consistent estimator (as $n\to\infty$).
See Figure~\ref{fig:Figure2} for a graphical illustration; and Algorithm~\ref{alg:Alg} for a formal description.

\begin{algorithm}
	\SetAlgoLined
	\caption{The proposed algorithm}
	\label{alg:Alg}
	\KwIn{Samples: $\{\by_i\}_{i=1}^n$; Radius: $R>0$.
	}
	
	Pick samples inside ball: $\m{K}=\left\{ i\in [n]\,:\quad \by_i\in \m{B}(\0,R) \right\}$. \\
	Form sample covariance:
	$
	\hbSigma = \sum_{i\in \m{K}} \by_i\by_i^\T \,.
	$		
	\Return{$\SpikeVecEst$}, principal eigenvector of $\hbSigma$.
\end{algorithm}

\begin{figure}[h]
	\vspace{.3in}
	\centering
	\includegraphics[width=\textwidth,height=0.2\textheight,keepaspectratio]{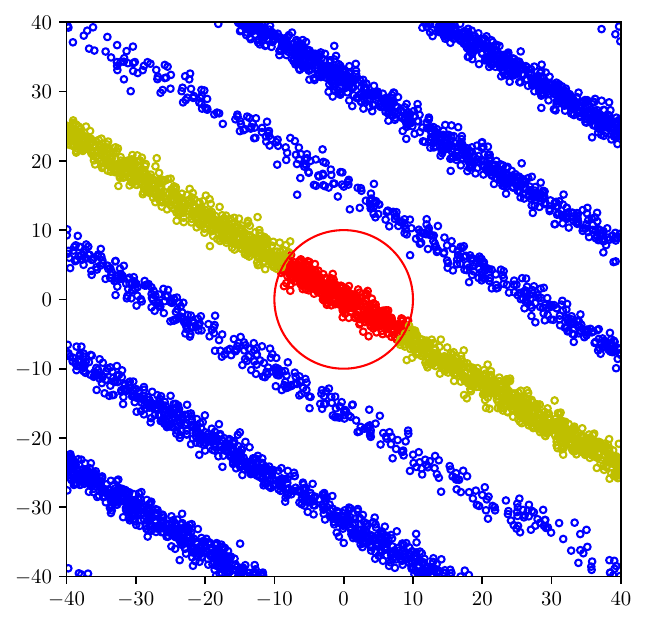}
	\vspace{.1in}
	\caption{An illustration of Algorithm~\ref{alg:Alg}. Red: points inside a small ball; Yellow: the central stripe (${\bx_i=\by_i}$).}
	\label{fig:Figure2}
\end{figure}

\paragraph{The truncation radius.} 
{In spite of its seeming simplicity, the behavior of Algorithm~\ref{alg:Alg} depends \emph{drastically} on the truncation radius $R$}. Its choice should {balance} between two opposing effects. On the one hand, the algorithm uses effectively ${\Expt|\Kpicked|=n\cdot \PBall}$ measurements for estimation (where ${\PBall:=\Pr(\bX\in \m{B}(\0,R)})$, so $R$ cannot be too small; on the other hand, we need to take only (or mostly) points from the central stripe, so $R$ cannot be too large. Let us start with an observation: when $R/\sqrt{k}<1$, $\PBall$ is exponentially small in $k$, so the algorithm requires $n\gtrsim 1/\PBall=\exp(\Omega(k))$ measurements. Consequently, to (potentially) overcome the curse of dimensionality one must set ${R/\sqrt{k}>1}$. In that case, $\PBall \approx 1\wedge(\sqrt{R/\Spike})$ (Lemma~\ref{lem:Pball}), so that for a large spike, $\Spike=\omega(R)$, $n\gtrsim \sqrt{\Spike/R}$; therefore, we shall henceforth restrict our attention to $\Spike=\poly(k)$. As we have said, $R$ cannot be chosen too large, and in general there is a rather delicate {\bf tradeoff} between the parameters $\Delta,R,\Spike$ and the direction $\SpikeVec\in\SphereK$ itself. Our main result, Theorem~\ref{thm:main}, says, roughly, the following: there \emph{is} a choice $R=\Theta(\sqrt{k})$ such that for \emph{most} directions, if $\Delta\gtrsim \sqrt{\log k}$ and $\Spike=\poly(k)$, then Algorithm~\ref{alg:Alg} estimates $\SpikeVec$ with small error from only $n=\poly(k)$ measurements.

\paragraph{On estimating $\Spike$.}
In this paper, we restrict our attention to estimating only the direction $\SpikeVec\in\SphereK$ (and not $\Spike$). We mention two potential strategies for estimating $\Spike$, not pursued here further due to space constraints:
\begin{itemize}
	\item 
	\cite{palombi2012numerical} studies covariance estimation from spherically-truncated Gaussian measurements. Relying on Proposition~\ref{obs:trunc-eigs}, they prove that the mapping $\bm{\lambda}\mapsto \bm{\mu}$ between the true and truncated eigenvalues is invertible, and propose a fixed point iteration to recover $\bm{\lambda}$ from $\bm{\mu}$ (given \emph{exactly}, without noise). Our proposed algorithm computes an estimate of $\bSigmaTruc$; computing error bounds for the method of \cite{palombi2012numerical}, applied to $\hbSigma$, is potentially challenging, especially in the regime where $\Spike\gg R^2=\Theta(k)$, where the mapping $\bm{\lambda}\mapsto \bm{\mu}$ is necessarily badly conditioned.
	\item 
	Since only one eigenvalue of $\bSigma$ is unknown, a more natural approach is to invert (numerically) the mapping $\Spike\mapsto \PBall(\Spike)=\Pr(\bX\in \m{B}(\0,R))$, which is strictly decreasing. 
\end{itemize}

\subsection{Main results}


The following is our main result. It shows that for most directions $\SpikeVec\in \SphereK$, the error attained by Algorithm~\ref{alg:Alg} can be made quite small, with $n$  reasonably controlled, and assuming only $\Delta \gtrsim \sqrt{\log k}$. To make the {presentation lighter}, we focus exclusively on the regime where the spike is not small, $\Spike=\Omega(1)$ (which is also practically more interesting). Below, $\SpikeVecEst$ denotes the largest eigenvector of $\hbSigma$, with the sign ambiguity resolved by assuming that $\langle \SpikeVec,\SpikeVecEst\rangle \ge 0$.


\begin{theorem}\label{thm:main}
	Fix a constant $M>12$, and set $R=\Theta(\sqrt{k})$ as in (\ref{eq:R-def}). There is a universal constant $\ConstDelta$ and a set $\m{U}_M\subseteq \SphereK$ with
	\[
	\Pr_{\SpikeVec\sim \Unif(\SphereK)} (\SpikeVec\in\m{U}_{M}) = 1-O_M(k^{-10}),
	\] 
	such that whenever $\Delta \ge \ConstDelta (M\sqrt{\log k}\vee \Spike^{\frac{1}{2(k-1)}})$ and $\SpikeVec\in \m{U}_M$, the following error bounds hold (depending on the magnitude of $\Spike$), with probability $1-O(k^{-10})$: 
	
	\begin{enumerate}
		\item Assume that $1\le \Spike \le k$ and 
		$n\gtrsim \log k$.
		Then
		\begin{equation}\label{eq:thm-small}
		\|\SpikeVec-\SpikeVecEst\| \lesssim \frac{1}{\sqrt{\Spike}} \left( \frac{k}{n} \vee \sqrt{\frac{k}{n} } \right) + \Spike^{-1} k^{-M^2+12} \,.
		\end{equation}
		\item Assume that $k \le \Spike \lesssim k^{2M^2-21}$.
		Then
		\begin{equation}\label{eq:thm-large}
		\|\SpikeVec-\SpikeVecEst\| \lesssim \frac{\sqrt{\Spike}}{n} + \sqrt{ \frac{1}{n}\sqrt{\frac{\Spike}{k}} } + \Spike^{1/2}k^{-M^2+10.5} \,.
		\end{equation}
	\end{enumerate}
\end{theorem}

\paragraph{Discussion.} Let us start with the small-spike regime, $1\le \Spike\le k$. The first term in Eq. (\ref{eq:thm-small}) is, up to prefactors, the error rate for PCA without modulo-folding, see e.g. \cite[Chapter 8]{wainwright2019high}; in particular, when $k=O(n)$, it matches the minimax lower bound Eq. (\ref{eq:minimax-rate}).  The second term corresponds to error incurred by erroneously taking ``bad'' points $\by_i$, that do not belong on the central stripe. By taking $\Delta\gtrsim \sqrt{\log k}$ large enough, this term can be made to decay arbitrarily (polynomially) fast as $k\to\infty$. We note that in this regime, the consequences of Theorem~\ref{thm:main} are, in fact, rather unsurprising: if $\SpikeVec\sim \Unif(\SphereK)$, then with high probability, ${\|\SpikeVec\|_{\infty}\lesssim \sqrt{\log k /k }}$. Since $\Delta\gtrsim \sqrt{\log k}$, the cube $[-\frac12 \Delta,\frac12\Delta)^k$ contains a segment $\{t\SpikeVec\,:\,t\in[-L,L]\}$ of length $2L\gtrsim \Delta/\|\SpikeVec\|_{\infty} \gtrsim \sqrt{k}$; consequently, a large fraction of $\bx_1,\ldots,\bx_n$ are actually themselves already inside the cube, since the ``typical length'' of the projection along $\SpikeVec$, $|\langle \SpikeVec,\bX\rangle|$, is $\lesssim \sqrt{\Spike}$ (the standard deviation). 

The ``interesting'' regime is {$\Spike\gg k$}. Note that unlike in the small-spike regime, here the error, Eq. (\ref{eq:thm-large}), \emph{increases} as $\Spike$ grows. Moreover, the magnitude of $\Spike$ has to be constrained by $\Delta$: $\Spike\lesssim k^{2M^2-21}$ (the constant $21$ is itself not particularly important, and can be improved). Thus, to retain the scaling $\Delta \sim \sqrt{\log k}$, $\Spike$ has to grow at most polynomially with $k$; in that case, note that the term $\Spike^{1/(2(k-1))}$ in the bound for $\Delta$ is always negligible. Furthermore, note that $n$ has to scale at least as $n\gtrsim \sqrt{\Spike}$, which anyhow precludes the practically of the algorithm when $\Spike$ is super-polynomial, regardless of the third term. 

Let us try to get some intuition for the particular form of the bounds (\ref{eq:thm-small}), (\ref{eq:thm-large}), by considering a simplified setting, where one had direct access to all the measurements $\bx_i$ that lie inside the ball $\m{B}(\0,R)$, and used them to perform PCA. There are roughly (Lemmas~\ref{lem:Nball}, \ref{lem:Pball})
\[
\tilde{n} \approx  \PBall n \approx (1\wedge \sqrt{k/\Spike})n
\]
such measurements. By Proposition~\ref{obs:trunc-eigs}, the population covariance $\bSigmaTruc$ is spiked, and one can show (Lemma~\ref{lem:spectral-gap}) that the effective spike is
\[
\tilde{\Spike} \approx k\wedge \Spike \,.
\]
(Note that $\bXTruc$ lies inside $\m{B}(\0,R)$, and consequently $\lambda_1(\bSigmaTruc)\le R^2=\Theta(k)$). In particular, note that when $\Spike=\omega(k)$, $\tilde{n}$ decreases with $\Spike$ but $\tilde{\Spike}$ cannot grow further to compensate for this; this is the reason why the error in Eq. (\ref{eq:thm-large}) degrades with $\Spike$. Now, assuming  that $\tilde{n}$ is large enough (this point is {a little subtle}, since we let $\tilde{\Spike}$ grow as well), the error is bounded like
\[
\|\SpikeVec-\SpikeVecEst\| \lesssim \frac{1}{\sqrt{\tilde{\Spike}}} \left( \frac{k}{\tilde{n}} \vee \sqrt{\frac{k}{\tilde{n}}} \right)\,,
\]
using ``standard'' bounds for PCA (e.g. \cite[Chapter 8]{wainwright2019high}). Plugging in the {above} estimates for $\tilde{n},\tilde{\Spike}$ recovers the first terms in Eqs. (\ref{eq:thm-small}), (\ref{eq:thm-large}). The challenging part of the analysis (and our main technical contribution) is to control the last term: namely, show that for most directions $\SpikeVec$, when $\Delta\gtrsim \sqrt{\log k}$, the contribution of the erroneously picked (``bad'') points is indeed very small with high probability.

\begin{figure}[h]
	\centering
	\begin{subfigure}[b]{0.4\textwidth}
		\centering
		\includegraphics[width=\textwidth]{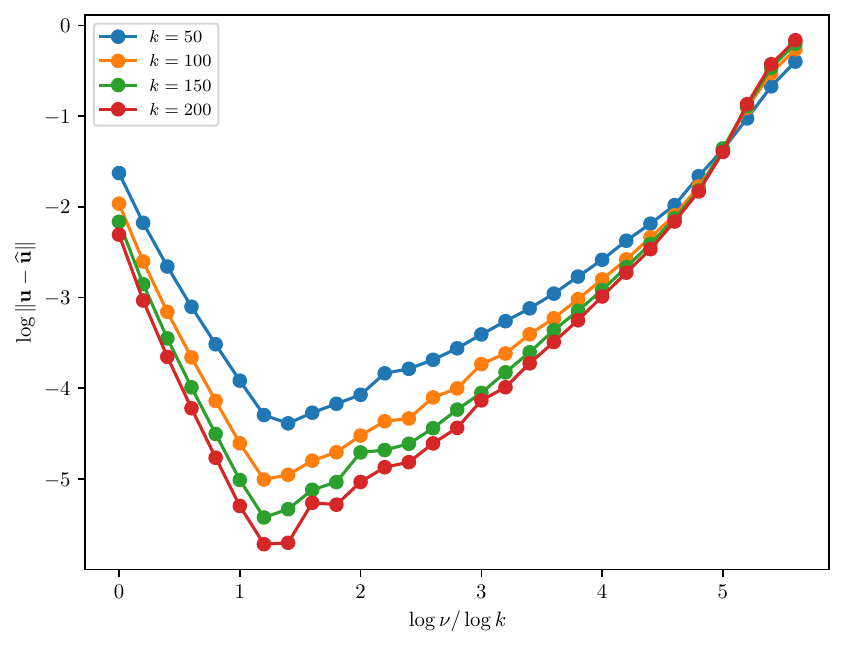}
		\caption{The estimation error, ${\|\SpikeVec-\SpikeVecEst\|}$, as $\Spike$ changes.}
		\label{fig:Experiment1}
	\end{subfigure}\\~\\
	\begin{subfigure}[b]{0.4\textwidth}
		\centering
		\includegraphics[width=\textwidth]{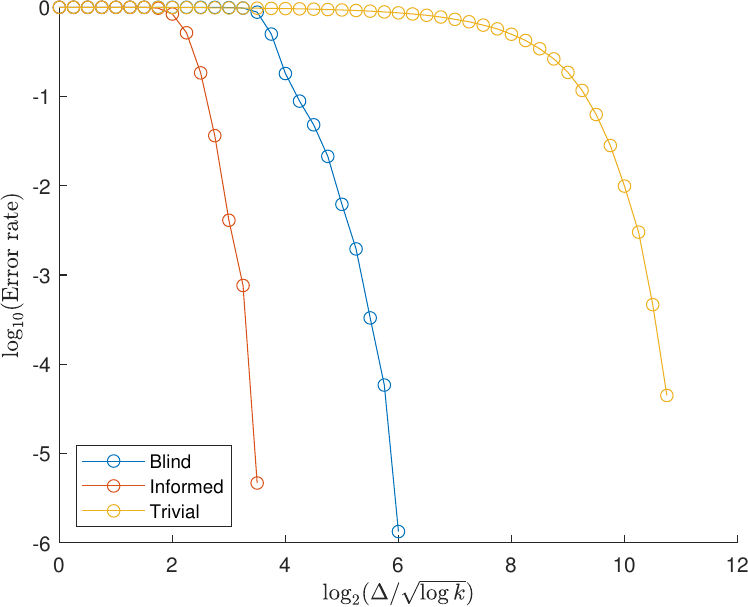}
		\caption{Performance of the informed vs. blind integer forcing decoder, as $\Delta/\sqrt{\log k}$ changes.}
		\label{fig:Experiment2}
	\end{subfigure}
\caption{Numerical results}
\end{figure}

\paragraph{Experiments.} 
We demonstrate the validity and relevance of our results through numerical experiments:
{
\begin{itemize}
	\item 
	In Figure~\ref{fig:Experiment1} we study the behavior of the error, ${\|\SpikeVec-\SpikeVecEst\|}$, as the spike magnitude $\Spike$ changes. For several values, $k=50,100,150,200$, we have set $\Delta=16\sqrt{\log k}$, $n=k^2$ and varied $\Spike=k^{\alpha}$ for an exponent $\alpha\in [0,5.6]$; each point on the graph is the average error across $T=2400$ repetition. We observe that as $k$ increases, scaling $\Delta\propto \sqrt{\log k}$ indeed suffices for estimation. Moreover, we see that for small spikes, $\alpha<1$, the error decreases as $\alpha$ increases, whereas when $\alpha>1$ the error increases; this is consistent with Theorem~\ref{thm:main}. 

	\item 
	In Figure~\ref{fig:Experiment2} we apply our algorithm as an intermediate step for blind unwrapping. We set $k=150$, $n=k^2\approx 2\cdot 10^4$, $\Spike=k^3\approx 3\cdot 10^6$ and vary $\Delta=2^{\delta}\sqrt{\log k}$. At every working point, we compute the error rate, namely the fraction of erroneously recovered samples $\widehat{p}_e=\frac1n \sum_{i=1}^n \Ind_{\bx_i\ne \widehat{\bx}_i}$, of the informed IF decoder (Eq. (\ref{eq:IF})), the blind IF decoder (computed from ${\hbSigma=\Spike \SpikeVecEst\SpikeVecEst+\bI}$), and the trivial decoder $\widehat{\bX}=\bY$. 
	For each method,
	$\delta$ is increased in jumps of $0.25$, until the point where $p_e\le 10^{-4}$; each point on the graph is the average of $T=200$ repetitions (so, overall, $nT\approx 10^6$ single recovery trials). We see, for this particular setup, a gap of around $\delta=2$ bits between the the informed and blind decoders, and of about $\delta=5$ bits between the blind and trivial decoders. To put things in context, a hypothetical quantization scheme based around modulo-folding and the blind decoder could save up to $5$ bits per coordinate (so $5k= 750$ bits overall) compared to the uniform quantizer (both designed so that the probability of a saturation is $\le 10^{-4}$).
	
\end{itemize} 
}

By Theorem~\ref{thm:main}, when $\Spike=\poly(k)$, the condition $\Delta\gtrsim \sqrt{\log k}$ ensure that one can estimate most directions $\SpikeVec\in\SphereK$ with $n=\poly(k)$ measurements. 
Recall that the present problem was motivated by the modulo-unfolding problem (which is a \emph{harder} problem). It turns out that for the latter, the condition $\Delta\gtrsim \sqrt{\log k}$ is actually necessary.
Thus, if one's goal is to solve the unwrapping problem (e.g. for implementing modulo-ADCs), and to that end estimates the covariance as an \emph{intermediate step}, then our algorithm succeeds with \emph{essentially} the smallest allowable dynamic range. We show the following (see Appendix~\ref{sec:proof-prop:lower-bound} for the proof):
\begin{proposition}
	\label{prop:lower-bound}
	Suppose that there exists $\widehat{\bX}=\widehat{\bX}(\bY)$, with $\Pr(\bX=\widehat{\bX})\ge 0.1$. Then ${\Delta\gtrsim (\sqrt{\log k}\vee \Spike^{\frac{1}{2k}})}$. 
\end{proposition} 

We remark in passing that, once we have obtained an estimate for $\SpikeVec$, and consequently for $\bSigma$, using Algorithm~\ref{alg:Alg}, we may use it to unwrap the measurements $\bx_1,\ldots,\bx_n$ (using, e.g., the IF decoder). Having unwrapped the samples, we can use standard methods (e.g., PCA) to get an improved estimate of $\SpikeVec$. 
We do not pursue this option here for two reasons: 1) The performance of such an algorithm depends on the unwrapping error probability, which is difficult to analyze. In particular, unwrapping errors could have a disastrous effect on the estimation error; 2) Our primary motivation for estimating $\SpikeVec$ in the first place was to perform unwrapping. To that end, once we obtain an estimate of $\SpikeVec$ with accuracy sufficient for unwrapping, further improvements are of limited interest. 

\section{Analysis}

In this section, we give a proof outline for our main result, Theorem~\ref{thm:main}. In the interest of space, the proofs of most technical lemmas are relegated to the Appendix.

For $\delta\in (0,1)$, set 
\begin{equation}\label{eq:z2-zinf}
\begin{split}
	z_2(\delta) &:= \left( k + 2\sqrt{k\log(1/\delta)} + 2\log(1/\delta) \right)^{1/2}\,,\quad \\
	z_\infty(\delta) &:= \sqrt{2\log k} + \sqrt{2\log (2/\delta)} \,,
\end{split}
\end{equation}
so that for $\bZ\sim \m{N}(\0,\Id_k)$ (see Appendix~\ref{sec:aux}, Lemma~\ref{lem:gaussian-tail-bounds}),
\begin{equation*}
	\begin{split}
		&\Pr(\|\bZ\|\ge z_2(\delta)) \le \delta,\quad
		\Pr(\|\bZ\|_\infty \ge z_\infty(\delta)) \le \delta \,.
	\end{split}
\end{equation*}

Going forward, we {\bf fix} a truncation radius:
\begin{equation}
	\label{eq:R-def}
	R = 2\sqrt{k} + z_2(0.1) = \Theta(\sqrt{k})\,.
\end{equation}
This {\it particular} choice is rather arbitrary. One could carry out the analysis with any $R=C\sqrt{k}$ for $C>1$; this would only change the constants in the bounds. 



\subsection{High level view} 

Divide the pairs $\{(\bx_i,\by_i)\}_{i\in [n]}$ into groups. Denote by
\begin{align*}
	&\Kpicked = \{ i \in [n]\,:\, \by_i \in \m{B}(\0,R)\}\,,\\
	&\Kball = \{ i \in [n] \,:\, \bx_i \in \m{B}(\0,R)\}\,,
\end{align*}
respectively the points that were picked by Algorithm~\ref{alg:Alg}, and those for which $\bx_i\in \m{B}(\0,R)$. Note that, conditioned on $i\in \Kball$, $\bx_i \overset{d}{=}\bXTruc$; hence $\{\bx_i\}_{i\in \Kball}$ is an i.i.d. sample from $\bXTruc$. Observe also that  $\Kball \subseteq \Kpicked$, since $\|\bX\bmod \Delta\|\le \|\bX\|$.
We denote by $\Kgood\subseteq \Kball$ the subset of measurements for which $\bx_i=\by_i$, in other words, such that $\bx_i\in [-\frac12 \Delta,\frac12 \Delta)^k$ to begin with. The measurements in $\Kbad=\Kpicked\setminus \Kgood$, will be called \emph{bad}. 
We have 
\begin{equation*}
\Kgood \subseteq \Kball \subseteq \Kpicked = \Kgood \sqcup \Kbad \subseteq [n] \,.
\end{equation*}
Now, the sample covariance,
\begin{equation}
	\label{eq:SigmaX}
	\hbSigmaX := \frac{1}{|\Kball|} \sum_{i\in\Kball}\bx_i \bx_i^\T
\end{equation}
is a consistent (as $|\Kball|\to \infty$) estimator for $\bSigmaTruc$, whose largest eigenvector is $\SpikeVec$ (Proposition~\ref{obs:trunc-eigs}). Consequently, PCA yields a consistent estimator for the unknown direction.
Alas, the set $\Kball$ is not directly observable, so the algorithm uses $\Kpicked\supseteq \Kball$ instead:\footnote{Note that we have normalized by $|\Kball|$, which is unknown. This is done {for the sake of} convenience in the analysis; the eigenvectors, of course, are not affected.  }
\begin{equation}
	\label{eq:SigmaY}
	\hbSigma = \frac{1}{|\Kball|}\sum_{i\in \Kpicked} \by_i\by_i^\T\,.
\end{equation}
This injects \emph{additional} error into the covariance estimation process, in two ways. First, the covariance is computed using the $\by_i$-s instead of the $\bx_i$-s (the latter are unknown); we have $\by_i=\bx_i$ only for $i\in \Kgood$, which may be a strict subset of $\Kpicked$. Second, we use additional samples, on top of $\Kball$: the points in $\Kpicked\setminus \Kball$ {necessarily} come from the wrong distribution. Set
\begin{equation*}
\begin{split}
\epsEst &:= \|\hbSigmaX-\bSigmaTruc\|\,,
 \quad \epsPick := \|\hbSigmaX-\hbSigma\| \,,
\end{split}
\end{equation*}
so that $\|\hbSigma - \bSigmaTruc\| \le \epsEst + \epsPick$.
A bound on this operator norm yields, by standard eigenvector perturbation results, a bound on $\|\SpikeVec-\SpikeVecEst\|$. 
Note: $\epsEst$ is simply the statistical estimation error in estimating $\bSigmaTruc=\Cov(\bXTruc)$ from $|\Kball|$ i.i.d. measurements; the other term, $\epsPick$, is the error induced through picking \emph{erroneous measurements}. We shall bound each error term separately.


\subsection{The covariance estimation error}

We start with $\epsEst$; the argument is quite standard. 
First, we show that with high probability, $|\Kball|$ is reasonably large. 
Denote
\begin{equation}
	\label{eq:Pball}
	\PBall = \PBall(\Spike;R) = \Pr\left( \bX\in \m{B}(\0,R)\right) \,,
\end{equation}
so that $|\Kball|\sim \mathrm{Binomial}\left( \PBall, n \right)$. Controlling $|\Kball|$ is straightforward using, e.g., Chernoff's inequality (Appendix~\ref{sec:aux}, Lemma~\ref{lem:chernoff}):
\begin{lemma}
	\label{lem:Nball} For any $\bu\in \SphereK$, for universal $c_1,c_2>0$,
	\[
	\Pr\left( |\Kball| \le c_1 \PBall n \,\big|\,\SpikeVec\right) \le 2e^{-c_2\PBall n} \,.
	\]
\end{lemma}

Next is an (tight, for large $\Spike$) estimate for $\PBall$; the (short) proof is relegated to Appendix~\ref{sec:proof-lem:Pball}:
\begin{lemma}\label{lem:Pball}
	\begin{equation*}
		0.9\,\erf\left( \sqrt{\frac{2k}{\Spike}}\right) \le \PBall \le \erf\left( \sqrt{\frac{2k + z_2(0.1)/2}{1+\Spike}}\right)\,,
	\end{equation*}
where, recall, the {error function} is defined by
\begin{equation}
	\erf(x) = \frac{2}{\sqrt{\pi}}\int_{0}^x  e^{-t^2}dt = \Pr(|\xi|\le \sqrt{2}x)\,.
\end{equation}
\end{lemma}

Note that Lemma~\ref{lem:Pball} implies that $\PBall=\Theta(1)$ when $\Spike=O(k)$, whereas $\PBall=\Theta(\sqrt{k/\Spike})$ when $\Spike=\omega(k)$; in other words, $\PBall = \Theta(1\wedge  \sqrt{k/\Spike})$.

Error bounds for covariance estimation rely on the concentration properties of the data. Thus, we need to show that $\bXTruc$ ``inherits'' the favorable properties of the underlying Gaussian vector $\bX$. 
We start with a general Lemma, whose proof appears in Appendix~\ref{sec:proof-lem:simple-convex}:
\begin{lemma}\label{lem:simple-convex}
	For every convex function $g:\RR^k\to \RR$,
	\[
	\Expt\left[ g(\bXTruc) \right] \le \Expt\left[ g(\bX) \right] \,.
	\]
\end{lemma}
The proof of Lemma~\ref{lem:simple-convex} relies on the Gaussian correlation inequality. As an important corollary, it allows us to control the sub-Gaussian and sub-exponential norms of 
$\bXTruc$; see Appendix~\ref{sec:proof-lem:epsEst-bound}, Lemma~\ref{lem:orlicz-norms-est}. 

Recall that by Proposition~\ref{obs:trunc-eigs}, $\bSigmaTruc$ is a spiked covariance matrix with largest eigenvector $\SpikeVec$:
\[
\lambda_1(\bSigmaTruc) > \lambda_2(\bSigmaTruc) = \ldots = \lambda_k(\bSigmaTruc)\,,
\]
so that applying Lemma~\ref{lem:simple-convex} (with $g(\bX)=\langle \bu_i,\bX\rangle^2$),
\[
\lambda_1(\bSigmaTruc)\le R^2\wedge (1+\Spike),\quad \lambda_2(\bSigmaTruc)\le 1\,.
\]
The rest of the analysis proceeds along rather standard lines, as in e.g. \cite[Section 8.2.2]{wainwright2019high}; the full details are given in Appendix~\ref{sec:proof-lem:epsEst-bound}. We prove:

\begin{lemma}\label{lem:epsEst-bound}
	Suppose that $\PBall n\gtrsim  \log k$. Then, with probability $1-O(k^{-10})$, 
	\begin{equation*}
		\begin{split}
		\epsEst
		&\lesssim 
		(k\wedge (1+\Spike))\sqrt{ \frac{\log k}{\PBall n} } \\
		&+ \sqrt{(k\wedge (1+\Spike)}\left( {\frac{k}{\PBall n}}\vee \sqrt{\frac{k}{\PBall n}} \right) \,.
		%
		\end{split}
	\end{equation*}
\end{lemma}

\subsection{The sample picking error}

%
Decompose
\begin{equation*}
	\sum_{i\in \Kpicked} \by_i\by_i^\T =   \sum_{i\in \Kgood} \bx_i\bx_i^\T + \sum_{i\in \Kbad} \by_i\by_i^\T \,,
\end{equation*}
so that
\begin{equation}\label{eq:epsPick}
	\begin{split}
		\epsPick
		&= \frac{1}{|\Kball|}\cdot \left\| \sum_{i\in \Kbad}\by_i\by_i^\T - \sum_{i\in \Kball\setminus \Kgood}\bx_i\bx_i^\T \right\| \\
		&\le \frac{2|\Kbad|\cdot R^2}{|\Kball|} 
		\lesssim  k\cdot \frac{|\Kbad|}{|\Kball|}\,.
	\end{split}
\end{equation}
Above, we used: $\|\bx_i\|\le R$ for $i\in \Kball$; $\|\by_i\|\le R$ for $i\in \Kpicked$; and
$
\left| \Kball\setminus\Kgood\right| \le \left| \Kpicked\setminus\Kgood\right| \le \left|\Kbad\right|
$.

The next Lemma is one of our main technical results. It states that for \emph{most} directions $\bu\in\SphereK$, the probabiliy that a pair $(\bX,\bY)$ is \emph{bad}, meaning that $\bY\in \m{B}(\0,R)$ but $\bX\ne \bY$, is overwhelmingly small provided that $\Delta\gtrsim \sqrt{\log k}\vee \Spike^{\frac{1}{2(k-1)}}$:
\begin{lemma}\label{lem:good-whp}
	Fix a constant $M\ge 1$. There is a universal $\ConstDelta>0$ and a subset $\m{U}_{M}\subseteq \SphereK$ with
	\[
	\Pr_{\bu\sim \Unif(\SphereK)}\left( \bu\in \m{U}_M \right) = 1 - O_M(k^{-10}) \,,
	\]
	such that if $\Delta \ge \ConstDelta (M\sqrt{\log k}\vee \Spike^{\frac{1}{2(k-1)}})$, then for all $\bu\in \m{U}_M$,
	\[
	\Pr\left( (\bX,\bY)\textnormal{ is bad }\,\big|\,\SpikeVec \right) \le k^{-M^2}\,.
	\]
\end{lemma}
The proof appears in Appendix~\ref{sect:proof-lem:good-whp}. The key idea is to reduce the problem into a question in geometric probability: whether a randomly rotated line segment is far away from all non-zero lattice points. 

Lemma~\ref{lem:good-whp} readily gives the following bound; the details are {given} in Appendix~\ref{sec:proof-lem:epsPick-bound}:
\begin{lemma}
	\label{lem:epsPick-bound}
	Assume the setup of Lemma~\ref{lem:good-whp}, with ${M>\sqrt{12}}$, $\SpikeVec\in\m{U}_M$, $\Delta\ge \ConstDelta (M\sqrt{\log k}\vee \Spike^{\frac{1}{2(k-1)}})$. Suppose that $\PBall n\gtrsim \log k$. With probability $1-O(k^{-10})$:
		\[
\epsPick \lesssim \frac{k^{-M^2+12}}{\PBall} \,.
		\]
%
\end{lemma}

\subsection{Concluding the analysis}

So far, we have shown that $\|\hbSigma-\bSigmaTruc\|$ is small with high probability. To deduce that their largest eigenvectors are close as well (using eigenvector perturbation results), we first need to show that the spectral gap of $\bSigmaTruc$ is large.
%
We prove the following in Appendix~\ref{sec:proof-lem:spectral-gap}:
\begin{lemma}\label{lem:spectral-gap}
	There are universal $C_1,C_2$ such that for $\Spike \ge e^{-C_1k}$,
	\[
	\lambda_1(\bSigmaTruc) \ge 1 + C_2(k \wedge \Spike) \,.
	\]
	Consequently, $\lambda_1(\bSigmaTruc)-\lambda_2(\bSigmaTruc)\gtrsim  (k \wedge \Spike)$.
\end{lemma}

The proof of Theorem~\ref{thm:main} follows by combining our bounds thus far. The details appear in Appendix~\ref{sec:proof-thm:main}.

\subsection*{Acknowledgements} This work was supported in part by ISF under Grant 1791/17 and in part by the GENESIS Consortium via the Israel Ministry of Economy and Industry. The work of Elad Romanov was supported in part by an Einstein-Kaye fellowship from the Hebrew University of Jerusalem.

\bibliographystyle{apalike}
\bibliography{refs}

\newpage

\onecolumn

\appendix

\section{Proof of Proposition~\ref{obs:trunc-eigs}, Item~\ref{item:obs1}}
\label{sec:proof-obs:trunc-eigs}

Decompose $\bX$ along the principal components:
\[
\bX = \sum_{i=1}^k \sqrt{\lambda_i}g_i \bu_i \,,
\]
where $g_1,\ldots,g_k\overset{i.i.d.}{\sim} \m{N}(0,1)$. Then,
\begin{align*}
	\bSigmaTruc 
	&= \Expt\left[ \bX\bX^\T \,\big|\, \|\bX\|^2\le R^2 \right] 
	= \sum_{i,j=1}^k \Expt\left[ \sqrt{\lambda_i\lambda_j }g_ig_j \bu_i\bu_j^\T \,\Big|\, \sum_{i=1}^k \lambda_i g_i^2 \le R^2 \right] \,.
\end{align*}
Now, observe that the cross terms, $i\ne j$, are zero, since conditioning onto the ball preserves the symmetry $(g_i,g_j)\mapsto (-g_i,g_j)$. Thus,
\begin{equation*}
	\bSigmaTruc = \sum_{i=1}^k \Expt\left[ \lambda_i g_i ^2\bu_i\bu_i^\T \,\Big|\, \sum_{i=1}^k \lambda_i g_i^2 \le R^2 \right] =: \sum_{i=1}^k \mu_i \bu_i\bu_i^\T \,,
\end{equation*}
and so the claim is proved. 

\section{Proof of Proposition~\ref{prop:lower-bound}}
\label{sec:proof-prop:lower-bound}

For brevity, define
\begin{eqnarray}
	p(\Spike;\SpikeVec) = \Pr_{\bX\sim \m{N}(\0,\Spike\SpikeVec\SpikeVec^\T+ \Id)} \left( \bX= \widehat{\bX}_{\mathrm{MAP}} \right)\,,
\end{eqnarray}
where $\bX_{\mathrm{MAP}}$ is the MAP estimator of $\bX$ from $\bY=[\bX]\bmod \Delta$; in other words, $p(\Spike;\SpikeVec)$ is the success probability of the MAP estimator at SNR $\Spike$ with spike direction $\SpikeVec$. Recalling that the MAP estimator is optimal in the sense of error probability, it is clear that to prove Proposition~\ref{prop:lower-bound}, it suffices to show that $p(\Spike;\SpikeVec) \ge 0.1$ implies that $\Delta\gtrsim \Spike^{1/2k}\vee \sqrt{\log k}$. 

We start with a simple observation:
\begin{lemma}\label{lem:obs-decreasing}
	The function $\Spike \mapsto p(\Spike;\SpikeVec)$ is decreasing. 
\end{lemma}
\begin{proof}
	For any $\Spike\ge 0$, denote $\bX_\Spike \sim \m{N}(\0,\Spike\SpikeVec\SpikeVec^\T+ \Id)$ and $\bY_\Spike = [\bX_\Spike]\bmod \Delta$, and let $g_\Spike : [-\Delta/2,\Delta/2)^k \to \RR^k$ be a deterministic function such that $g_\Spike(\bY_\Spike)$ is the MAP estimator for $\bX_\Spike$. Fix any $\tau\ge 0$; we shall now construct a suboptimal estimator for $\bX_\Spike$ given $\bY_\Spike$, based on $g_{\Spike+\tau}(\cdot)$. The idea is simple: we generate known noise $\bN\sim \m{N}(\0,\tau\SpikeVec\SpikeVec^\T)$ and set $\bY'=[\bY_\Spike + \bN]\bmod \Delta$, which also equals $\bY'=[\bX_\Spike+\bN]\bmod \Delta$. Note that $(\bX_{\Spike+\tau},\bY_{\Spike+\tau}) \overset{d}{=} (\bX_{\Spike}+\bN,\bY')$. Considering the sub-optimal estimator $\widehat{\bX}'=g_{\Spike+\tau}(\bY')-\bN$ for $\bX_\Spike$, we conclude,
	\[
	p({\Spike;\SpikeVec}) \ge 
	\Pr(\bX_\Spike= \widehat{\bX}') =
	\Pr\left( \bX_\Spike = g_{\Spike+\tau}(\bY')-\bN\right) = \Pr\left( \bX_{\Spike+\tau} = g_{\Spike+\tau}(\bY_{\Spike+\tau})\right) = p({\Spike+\tau;\SpikeVec}) \,.
	\]
\end{proof}

Let us start by showing $\Delta\gtrsim \sqrt{\log k}$. By Lemma~\ref{lem:obs-decreasing}, the assumptions of Proposition~\ref{prop:lower-bound} imply that $p(0;\SpikeVec)\ge 0.1$. Now, it is easy to see that when $\Spike=0$, $\widehat{\bX}_{\mathrm{MAP}}=\bY$; in this case, the problem simply decouples across the different coordinates. Thus, 
\[
0.1 \le p(0;\SpikeVec) = \Pr(\bX=\bY)=\Pr\left( |X_i|\le \Delta/2\,\textrm{ for all }1\le i \le k)\right) = \left\{\erf\left(\Delta/2^{3/2}\right)\right\}^k \,,
\]
and therefore $\erf(\Delta/2^{3/2})\ge 1-ck^{-1}$ for some universal $c\ge 0$. Clearly, then, $\Delta=\omega(1)$ for large $k$, so by the standard estimate $1-\erf(x)\gtrsim e^{-x^2}/x$ (for large $x$), we get $e^{-c_1\Delta^2}/\Delta \lesssim 1/k$ hence $\Delta \gtrsim \sqrt{\log k}$.

It remains to show $\Delta\gtrsim \Spike^{1/2k}$. To that end, we will use a simple geometric characterization of the MAP estimator, following \cite{romanov2021blind}. Let 
\[
\m{L} = \Delta\bSigma^{-1/2}\ZZ^k = \left\{ \Delta\bSigma^{-1/2}\bt \, : \, \bt\in \ZZ^k \right\}
\]
be the lattice generated by the matrix $\Delta\bSigma^{-1/2}$, and denote by $V_0 \subset \RR^k$ the Voronoi cell of $\0\in \m{L}$ (that is, all points $\ba\in \RR^k$ whose closest lattice point is $\0\in \m{L}$). 
By \cite[Section III, Eq. (38)]{romanov2021blind}, the success probability of the MAP estimator is
\[
\Pr\left( \widehat{\bX}^{\mathrm{MAP}}(\bY) = \bX \right) = \Pr_{\bZ\sim \m{N}(0,\bI_k)}\left( \bZ \in V_0 \right) \,.
\]
Now, it is a well-known fact that $V_0$ is a convex symmetric set, with $\mathrm{vol}_{k}(V_0)=|\Delta\bSigma^{-1/2}|=\Delta^{k}/|\bSigma|^{1/2}$. Let $r_0$ be the \emph{effective radius} of $\m{L}$, defined by 
\[
\mathrm{vol}_k(\m{B}(\0,r_0))=\mathrm{vol}_{k}(V_0) \Longrightarrow r_0 = \frac{\Delta}{|\bSigma|^{\frac{1}{2k}}\VolK^{1/k}}\,,
\]
($\VolK$ denotes the volume of the Euclidean unit ball). 
Recall that among all convex bodies with a given (finite) volume, a ball has the largest Gaussian measure. Thus,
\[
0.1 \le \Pr_{\bZ\sim \m{N}(0,\bI_k)}\left( \bZ \in V_0 \right) \le \Pr_{\bZ\sim \m{N}(0,\bI_k)}\left( \bZ \in \m{B}(\0,r_0) \right) = \Pr_{\bZ\sim \m{N}(0,\bI_k)}\left( \|\bZ\|^2 \le r_0^2 \right) \,.
\]
Note that $\|\bZ\|^2\sim \chi^2_{(k)}$, which concentrates around $k$ with ``typical'' deviations of order $O(\sqrt{k})$ (see, e.g., Lemma~\ref{lem:gaussian-tail-bounds}). This gives $r_0\gtrsim \sqrt{k}$, so
\[
\Delta \gtrsim  |\bSigma|^{\frac{1}{2k}} \left( \VolK^{1/k}\sqrt{k}\right) = (1+\Spike)^{1/2k}\left( \VolK^{1/k}\sqrt{k}\right) \gtrsim (1+\Spike)^{1/2k}\,,
\]
where the last inequality follows from Stirling's approximation: $\VolK \overset{k\to\infty}{\approx} \frac{1}{\sqrt{k \pi}}\left( \frac{2\pi e}{k}\right)^{k/2}$, and therefore ${\VolK^{1/k}\sqrt{k}=\sqrt{2\pi e} + o(1) = \Theta(1)}$.

\section{Proof of Lemma~\ref{lem:epsEst-bound}}
\label{sec:proof-lem:epsEst-bound}
	
Decompose $\bXTruc$ along the principal components:
\[
\bXTruc = w\bu + \bp\,,
\]
where $w$ is the projection along $\bu$ and $\bp$ is the orthogonal complement. Note that, while $w$ and $\bp$ are uncorrelated, they are \textbf{not} independent (as was the case without truncation, for a Gaussian vector) since we condition on $w^2+\|\bp\|^2\le R^2$. Also, recalling the ``spiky'' structure of $\bSigmaTruc$,
\begin{equation}
\Expt(w^2)=\lambda_1(\bSigmaTruc),\quad \Expt[\bp\bp^\T] = \bSigmaTruc - \lambda_1(\bSigmaTruc)\cdot \SpikeVec\SpikeVec^\T = \lambda_2(\bSigmaTruc)\cdot (\Id - \SpikeVec\SpikeVec^\T) \,.
\end{equation}

Condition on $N=|\Kball|$, and denote for convenience $\Kball=\{1,\ldots,N\}$, so that $\{\bx_i\}_{i=1}^N$ are i.i.d. measurements from $\bXTruc$. Write
\begin{equation}
	\label{eq:hbSigmaX-decomposition}\begin{split}
			\hbSigmaX 
		&= \frac{1}{N} \sum_{i=1}^N \left( w_i\bu + \bp_i \right)\left( w_i\bu+\bp_i \right)^\T \\
		&= \left( \frac1N\sum_{i=1}^N w_i^2 \right)\bu\bu^\T + \left(\frac1N \sum_{i=1}^N w_i\bp_i\right)\bu^\T  + \left(\frac1N\left(\sum_{i=1}^N w_i\bp_i\right)\bu^\T\right)^\T + \frac1N \sum_{i=1}^N \bp_i\bp_i^\T \,,
	\end{split}
\end{equation}
so that the error can be decomposed as $\left\| \hbSigmaX-\bSigmaTruc \right\|\le \eps_1 + 2\eps_2 + \eps_3$, with
\begin{equation}\label{eq:hbSigmaX-eps}
	\begin{split}
		\eps_1 &:= \left| \frac1N\sum_{i=1}^N w_i^2 - \Expt[w^2] \right|\,,\\
		\eps_2 &:= \left\| \frac1N \sum_{i=1}^N w_i\bp_i \right\| \,, \\
		\eps_3 &:= \left\| \frac1N \sum_{i=1}^N \bp_i\bp_i^\T - \Expt\left[\bp\bp^\T\right] \right\| \,.
	\end{split}
\end{equation}

We first show that $\bXTruc$ inherits the sub-Gaussian concentration properties of $\bX$. 
We denote, respectively, the sub-Gaussian and sub-exponential norms by $\|\cdot\|_{\psi_2}$ and $\|\cdot \|_{\psi_1}$. For a quick reminder on these norm (and Orlicz norms in general), see Definition~\ref{def:orlicz-norm} and Lemma~\ref{lem:orlicz-props}.

\begin{lemma}\label{lem:orlicz-norms-est}
	We have
	\[
\|\bp\|_{\psi_2} \lesssim 1\,,\quad 	\|w^2\|_{\psi_1} \lesssim k\wedge (1+\Spike)\,.
	\]
	and
	\[
	\left\|w \bp \right\|_{\psi_1} \lesssim \sqrt{k\wedge (1+\Spike)}\,.
	\]
\end{lemma}
\begin{proof}
	Let $\psi : [0,\infty) \to [0,\infty)$ be convex and increasing, and {let $g:\RR\to\RR$ be such that $\bx\mapsto |g(\bx)|$ is convex.}
	Observe that $x\mapsto \psi(|g(x)|)$ is convex, and consequently, by Lemma~\ref{lem:simple-convex}, $\|g(\bXTruc)\|_{\psi} \le \|g(\bX)\|_{\psi}$, where $\|\cdot \|_{\psi}$ is the Orlicz $\psi$-norm (see Definition~\ref{def:orlicz-norm}). {Consequently},
	\[
	\|w^2\|_{\psi_1} \le \|\langle \bu,\bX\rangle^2\|_{\psi_1} \lesssim 1+\Spike,\quad \|\bp\|_{\psi_2} \le \left\| (\bI-\bu\bu^\T)\bX \right\|_{\psi_2} \lesssim 1 \,.
	\]
	{Furthermore, } using Lemma~\ref{lem:orlicz-props}, Items \ref{item:orlicz1} and \ref{item:orlicz5}, 
{
	\[
	\|w^2\|_{\psi_1} = \|w\|_{\psi_2}^2 \lesssim \|w\|^2_{\infty} \le R^2 \lesssim k \,.
	\]}
	This proves the first two bound. As for the last one,
{	
	\[
	\|w\bp\|_{\psi_1} = \sup_{\bv\in\SphereK} \|w\langle \bv,\bp\rangle \|_{\psi_1} 
	\overset{(\star)}{\le} \|w\|_{\psi_2}\cdot \sup_{\bv\in\SphereK}\|\langle \bv,\bp\rangle\|_{\psi_2} = \|w\|_{\psi_2}\|\bp\|_{\psi_2} \overset{(\star \star)}{\lesssim} \sqrt{k\wedge (1+\Spike)} \,,
	\]
	where $(\star)$ follows from Lemma~\ref{lem:orlicz-props}, Item \ref{item:orlicz2}, and $(\star \star)$ follows from Lemma~\ref{lem:orlicz-props}, Item~\ref{item:orlicz1}, and the first part of this proof. 
}
\end{proof}

We now bound the errors $\eps_1,\eps_2,\eps_3$, again conditioned on $N=|\Kball|$:

\begin{lemma}\label{lem:eps1}
	Assume that $N\gtrsim \log k$. Then, with probability $1-O(k^{-10})$, 
	\[
	\eps_1 \lesssim (k\wedge (1+\Spike)) \sqrt{\frac{\log k}{N}} \,.
	\]
\end{lemma}
\begin{proof}
	By the centralization Lemma (Lemma~\ref{lem:orlicz-props}, Item~\ref{item:orlicz3}) and Lemma~\ref{lem:orlicz-norms-est},
	\[
	\left\|w_i^2-\Expt\left[w_i^2\right]\right\|_{\psi_1}\le \|w_i^2\|_{\psi_1} \lesssim k\wedge (1+\Spike) \,.
	\]
	By Bernstein's inequality (Lemma~\ref{lem:bernstein}),
	\[
	\Pr\left(\eps_1 \ge t\right) \le 2\exp \left[ -c_1N (\delta\wedge \delta^2)\right],\quad \delta := \frac{t}{k\wedge (1+\Spike)} \,.
	\]
	Set $t=\sqrt{\frac{{10}}{c_1}}(k\wedge(1+\Spike))\sqrt{\frac{\log k}{N}}$. Then whenever $N \ge \frac{10}{c_1}\log k$, the probability is $\le 2k^{-10}$.
\end{proof}

\begin{lemma}\label{lem:eps2}
	With probability $1-2e^{-\Omega(k)}$,
	\[
	\eps_2 \lesssim  \sqrt{k\wedge (1+\Spike)} \left(\frac{k}{N}\vee \sqrt{\frac{k}{N}}\right) \,.
	\]
\end{lemma}
\begin{proof}
	Set $\bq = \frac1N\sum_{i=1}^N w_i\bp_i$, and observe that $\Expt[\bq]=\0$, since $w_i$ and $\bp_i$ are uncorrelated. We want to bound $\|\bq\|$ with high probability; to that end, we use a standard $\eps$-net argument, executed in detail for the sake of completeness. Using \cite[Corollary 4.2.13]{vershynin2018high}, {fix} a $1/2$-net $\m{N}$ of $\SphereK$ of size $|\m{N}|\le 5^k$. Let $\tilde{\bv}\in\m{N}$ {be a member of the net, such that $\left\|\frac{\bq}{\|\bq\|}-\tilde{\bv}\right\|\le 1/2$. Now,
	\begin{align*}
\|\bq\| 
= \left\langle \bq,\frac{\bq}{\|\bq\|}\right\rangle 
= \left\langle \bq,\frac{\bq}{\|\bq\|}-\tilde{\bv}\right\rangle + \left\langle \bq,\tilde{\bv}\right\rangle 
\le  \|\bq\| \left\|\frac{\bq}{\|\bq\|} -\tilde{\bv}\right\| + \langle \bq,\tilde{\bv}\rangle 
\le \frac12 \|\bq\| + \langle \bq,\tilde{\bv}\rangle\,,
	\end{align*}
which implies $\|\bq\|\le 2\langle \bq,\tilde{\bv}\rangle$. }
	Consequently, $\|\bq\|\le 2\max_{\bv\in\m{N}} \langle \bq,\bv\rangle$, {so it suffices to bound the latter.} Recalling, by Lemma~\ref{lem:orlicz-norms-est}, that ${\|w_i\bp_i\|_{\psi_1}\lesssim \sqrt{k\wedge (1+\Spike)}}$, by Bernstein's inequality and a union bound over the net,
	\[
	\Pr\left( \max_{\bv\in\m{N}} \langle \bq,\bv\rangle \ge t \right) \le 2\cdot 5^k \cdot \exp \left[ -c_1N (\delta\wedge \delta^2)\right],\quad \delta := \frac{t}{\sqrt{k\wedge (1+\Spike)}}\,.
	\] 
	Set 
	\[
	t = \sqrt{k\wedge (1+\Spike)} \left( \frac{10}{c_1}\frac{k}{N} \vee \sqrt{\frac{10}{c_1}\frac{k}{N}} \right) \,,
	\]
	so that the probability is bounded by $2\cdot 5^k e^{-10k}=2e^{-\Omega(k)}$.
\end{proof}

\begin{lemma}\label{lem:eps3}
	With probability $1-2e^{-\Omega(k)}$,
	\[
	\eps_3 \lesssim \left(\frac{k}{N}\vee \sqrt{\frac{k}{N}}\right) \,.
	\]
\end{lemma}
\begin{proof}
	By Lemma~\ref{lem:orlicz-norms-est}, the vectors $\bp_i$ are $O(1)$-sub-Gaussian. 
	The claim follows by Lemma~\ref{lem:vershynin}, applied with $t\approx \sqrt{k}$, $\delta\approx \sqrt{\frac{k}{N}}$.
\end{proof}

\paragraph{Proof of Lemma~\ref{lem:epsEst-bound}.} The proof follows from Lemmas~\ref{lem:eps1}, \ref{lem:eps2} and \ref{lem:eps3}, combined with $N\gtrsim \PBall n$ from Lemma~\ref{lem:Nball}.

\section{Proof of Lemma~\ref{lem:good-whp}}
\label{sect:proof-lem:good-whp}

The core of the argument is this: we reduce the question of whether ${\Pr( (\bX,\bY)\textnormal{ is bad}\,|\,\bu)}$ is large to a \emph{geometric question}; specifically, whether a randomly rotated line segment is close to any non-zero lattice point.
The details proceed as follow.


Recall: the pair $(\bX,\bY)$ is bad when $\bY\in \m{B}(\0,R)$ but $\bX\notin \Cube := [-\frac12 \Delta,\frac12 \Delta)^k$. Our goal is to show that for most directions $\SpikeVec\in \SphereK$, the probability that $(\bX,\bY)$ is bad is small, specifically,
\[
\Pr\left( (\bX,\bY)\textnormal{ is bad} \,|\,\SpikeVec\right) \le k^{-M^2} \,.
\]

We start by constraining ourselves to a set of ``typical'' vectors $\bX$. 
As in Eq. (\ref{eq:X}), write, $\bX=\sqrt{\Spike}\xi\bu + \bZ$, for independent $\xi\sim \m{N}(0,1)$, $\bZ\sim \m{N}(\0,\bI)$. Let $\delta>0$ be a confidence parameter (we shall set $\delta=k^{-M^2}$ later), and consider the event
\begin{equation}\label{eq:EventX}
	\m{E}_{\bX} = \left\{ \|\bZ\|_2 \le z_2(\delta/3),\, \|\bZ\|_\infty \le z_\infty(\delta/3),\,|\xi|\le h(\delta/3) \right\}
\end{equation}
where $h(\delta)=\sqrt{2\log (2/\delta)}$ is such that $\Pr(|\xi|\ge h(\delta))\le \delta$,  and $z_2(\delta),z_{\infty}(\delta)$ are as in Eq. (\ref{eq:z2-zinf}). 
Clearly, $\Pr(\m{E}_{\bX})\le \delta$. 

Operating under $\m{E}_{\bX}$, let us bound the event $\{(\bX,\bY)\textnormal{ is bad} \}$ by another, larger, event. To start, note that $\bX\ne \bY$ implies that $\bX=\bY+\Delta \bt$ for some \emph{non-zero} lattice vector $\bt\in \ZZ^k\setminus\{\0\}$. Consequently, when $\bY\in \m{B}(\0,R)$, $\bX\ne \bY$ implies that $\bX \in \bigcup_{\bt\in \ZZ^k\setminus\{\0\}} \m{B}\left( \Delta\bt,R \right)$. Decomposing $\bX$, this further implies that 
\begin{equation}\label{eq:helper1}
\sqrt{\Spike}\xi \SpikeVec \in \bigcup_{\bt\in \ZZ^k\setminus\{\0\}} \m{B}\left( \Delta\bt,R + \|\bZ\| \right) \,.
\end{equation}
As for the condition $\bX\notin \Cube$, equivalently $\|\bX\|_{\infty}>\frac12 \Delta$, it follows from the triangle inequality that 
\begin{equation}\label{eq:helper2}
	\|\sqrt{\Spike}\xi \SpikeVec\|_\infty = |\sqrt{\Spike}\xi| \|\SpikeVec\|_{\infty} \ge \frac12 \Delta-\|\bZ\|_{\infty} \,.
\end{equation}
Let $\m{U}_{1,M}\subseteq \SphereK$ be the set of \emph{incoherent} directions, 
\begin{equation}
	\m{U}_{1,M} := \left\{ \SpikeVec\in \SphereK\,:\, \|\SpikeVec\|_{\infty} \le \frac{1}{C_1}\sqrt{\frac{\log k}{k}}\right\}\,,
\end{equation} 
with $C_1$ a universal constant such that $\Pr(\SpikeVec\in \m{U}_{1,M})=1-O(k^{-10})$ for $\SpikeVec\sim \Unif(\SphereK)$ (see Lemma~\ref{lem:unif-inf}). Now, under $\m{E}_{\bX}$, and assuming that $\SpikeVec\in \m{U}_{1,M}$, Eqs. (\ref{eq:helper1}) and (\ref{eq:helper2}) imply that 
\begin{equation}\label{eq:helper3}
	\sqrt{\Spike}\xi \SpikeVec \in \bigcup_{\bt\in \ZZ^k\setminus\{\0\}} \m{B}\left( \Delta\bt,R_\delta \right),\quad 
	|\sqrt{\Spike}\xi| \ge C_1 \Delta_\delta \sqrt{\frac{k}{\log k}}\,,
\end{equation}
where we set
\begin{equation}\label{eq:def-Delta-R-delta}
	\Delta_\delta := \frac12 \Delta - z_{\infty}(\delta/3),\quad R_\delta := R + z_2(\delta/3) \,.
\end{equation}
Henceforth, we shall assume $\Delta$ to be large enough so that $\Delta_\delta>0$. Consider the line segment, $L(\SpikeVec)\subset \RR^k$,
\begin{equation}\label{eq:Line-tilde}
	{L}(\SpikeVec) = \{s\SpikeVec\quad :\quad A \le s \le B\}\,,\quad\textrm{where}\quad A := C_1\Delta_{\delta} \sqrt{\frac{k}{\log k}},\,\,B:=\sqrt{\Spike}h(\delta/3)=\sqrt{2\Spike\log(6/\delta)} \,.
\end{equation}
Observe that under $\m{E}_{\bX}$, the occurrence of the event in Eq. (\ref{eq:helper3}) implies, in particular, that 
\[
L(\SpikeVec) \cap \bigcup_{\bt\in \ZZ^k\setminus\{\0\}} \m{B}\left( \Delta\bt,R_\delta \right) \ne \emptyset \,.
\]
Note that given $\SpikeVec$, this is a \emph{deterministic} geometric condition. Set 
\begin{equation}
	\m{U}_{2,M} := \left\{ \SpikeVec\in\SphereK\,: L(\SpikeVec) \cap \bigcup_{\bt\in \ZZ^k\setminus\{\0\}} \m{B}\left( \Delta\bt,R_\delta \right) = \emptyset \right\} \,.
\end{equation}
and
\begin{equation}
	\m{U}_{M} := \m{U}_{1,M} \cap \m{U}_{2,M} \,.
\end{equation}
Summarizing the preceding discussion, we have argued that whenever $\SpikeVec\in \m{U}_{M}$, the event $\m{E}_{\bX}$ already implies that $(\bX,\bY)$ are \emph{good}. Thus, for $\SpikeVec\in\m{U}_{M}$, 
\begin{equation}
\Pr\left( (\bX,\bY)\textnormal{ is bad} \,\big|\, \SpikeVec\right) \le \Pr(\m{E}_{\bX}^c\,|\,\SpikeVec)  = \Pr(\m{E}_{\bX}^c) \le \delta \,.
\end{equation}

The proof of Lemma~\ref{lem:good-whp} will follow from the following auxiliary result:

\begin{lemma}\label{lem:good-whp-helper}
	Fix a constant $M\ge 1$ and set $\delta=k^{-M^2}$. There is a universal constant $\ConstDelta>0$,  
	such that if $\Delta \ge \ConstDelta \left( M\sqrt{\log k}\vee \Spike^{\frac{1}{2(k-1)}} \right)$ then, for $\bu\sim \Unif(\SphereK)$, 
	\[
	\Pr(\SpikeVec\notin \m{U}_{2,M}) \le O_M(k^{-10}) \,.
	\]
\end{lemma}

Lemma~\ref{lem:good-whp-helper} is purely a result in geometric probability. It states the following: take the 1D line segment $\tilde{L}(\be_1) \subseteq \RR^k$, and rotate it uniformly in space (apply a random rotation $U\sim \mathrm{Haar}(O(k))$). Then with high probability, the rotated segment will end up far away from all non-zero lattice points.
The remainder of this section is devoted to proving Lemma~\ref{lem:good-whp-helper}. 

Let us discretize the interval $[A,B]$ into disjoint sub-intervals of maximal length, such that the length of a sub-interval is $\le R_\delta$; let $s_0=A<s_1<\ldots<s_T=B$ be the corresponding end-points, and note that we may take $T\le \left\lceil \frac{B-A}{R_\delta} \right\rceil+1$. 
Clearly, any point in ${L}(\bu)$ must be $0.5R_\delta$-close to some point in $\{s_1\bu,\ldots,s_T\bu\}$. In particular, ${L}(\bu) \cap  \bigcup_{\bt\in \ZZ^k\setminus\{\0\}}\m{B}(\Delta\bt,R_\delta) \ne \emptyset$ implies that $s_\ell \bu \in \bigcup_{\bt\in \ZZ^k\setminus\{\0\}}\m{B}(\Delta\bt,1.5R_\delta)$ for some $1\le \ell \le T$. Consequently, 
\begin{equation}\label{eq:good-helper-bound1}
	\begin{split}
		\Pr\left( {L}(\bu) \cap  \bigcup_{\bt\in \ZZ^k\setminus\{\0\}}\m{B}(\Delta\bt,R_\delta) \ne \emptyset\right) 
		&\le \Pr\left( \{s_1\bu,\ldots,s_T\bu\} \cap  \bigcup_{\bt\in \ZZ^k\setminus\{\0\}}\m{B}(\Delta\bt,1.5R_\delta) \ne \emptyset\right)  \\
		&\le \sum_{\ell=1}^T \Pr\left( s_\ell\bu \in \bigcup_{\bt\in \ZZ^k\setminus\{\0\}}\m{B}(\Delta\bt,1.5R_\delta) \right) \\
		&= \sum_{\ell=1}^T \Pr\left( \bu \in \bigcup_{\bt\in \ZZ^k\setminus\{\0\}}\m{B}\left(\frac{\Delta}{s_\ell}\bt,\frac{1.5R_\delta}{s_\ell}\right) \right) \\
		&=: \sum_{\ell=1}^T p_\ell \,.
	\end{split}
\end{equation}
Since $\bu\sim \Unif(\SphereK)$, each term of Eq. (\ref{eq:good-helper-bound1}) is, by definition,
\begin{equation*}
	\begin{split}
		p_\ell 
		= \frac{\sigma_{k-1}\left( \SphereK \cap \bigcup_{\bt\in \ZZ^k\setminus\{\0\}}\m{B}\left(\frac{\Delta}{s_\ell}\bt,\frac{1.5R_\delta}{s_\ell}\right) \right)}{\sigma_{k-1}(\SphereK)} 
		\le \frac{\sum_{\bt\in \ZZ^k\setminus\{\0\}} \sigma_{k-1}\left( \SphereK \cap \m{B}\left(\frac{\Delta}{s_\ell}\bt,\frac{1.5R_\delta}{s_\ell}\right)\right) }{\sigma_{k-1}(\SphereK)} \,,
	\end{split}
\end{equation*}
where $\sigma_{k-1}(\cdot)$ denotes the surface area. Note that, one the one hand,
\begin{align*}
	\sigma_{k-1}\left( \SphereK \cap \m{B}\left(\frac{\Delta}{s_\ell}\bt,\frac{1.5R_\delta}{s_\ell}\right)\right) 
	&\le \sigma_{k-1}\left(  \partial\left( \m{B}(\0,1) \cap \m{B}\left(\frac{\Delta}{s_\ell}\bt,\frac{1.5R_\delta}{s_\ell}\right)  \right) \right) \\
	&\overset{(\star)}{\le} \sigma_{k-1}\left( \partial \m{B}\left(\frac{\Delta}{s_\ell}\bt,\frac{1.5R_\delta}{s_\ell}\right)\right) \\
	&= \left(\frac{1.5R_\delta}{s_\ell}\right)^{k-1}\sigma_{k-1}(\SphereK)\,,
\end{align*}
where $\partial(\cdot)$ denotes the boundary of a set, and $(\star)$ follows from the well-known fact that for convex bodies $L\subset K$, $\sigma_{k-1}(\partial L) \le \sigma_{k-1}(\partial K)$; see, e.g., \cite[Theorem B.1.14]{artstein2015asymptotic}.
On the other hand, clearly, $\sigma_{k-1}\left( \SphereK \cap \m{B}\left(\frac{\Delta}{s_\ell}\bt,\frac{1.5R_\delta}{s_\ell}\right)\right)=0$ whenever $\SphereK \cap \m{B}\left(\frac{\Delta}{s_\ell}\bt,\frac{1.5R_\delta}{s_\ell}\right)=\emptyset$.
Setting
\begin{equation}
	N_\ell = \left| \left\{ \bt\in \ZZ^k\setminus \{\0\}\,:\, \SphereK \cap \m{B}\left(\frac{\Delta}{s_\ell}\bt,\frac{1.5R_M}{s_\ell}\right)\ne \emptyset \right\} \right|\,,
\end{equation}
we conclude that
\begin{equation}
	\label{eq:helper-pell-bound}
	p_{\ell} \le N_\ell \left(\frac{1.5R_\delta}{s_\ell}\right)^{k-1}\,.
\end{equation}
%
%
\begin{lemma}
	\label{lem:packing-bound}
	We have 
	\[
	N_\ell \le  
	\begin{cases}
		\VolK \cdot \left(\frac{s_\ell + 1.5R_\delta }{\Delta} + \sqrt{k}\right)^k  \quad&\textrm{if } s_\ell<  \Delta\sqrt{k}+1.5R_\delta \\
		k\cdot \VolK \cdot \left( \frac{3R_\delta}{\Delta} + 2\sqrt{k} \right)\left( \frac{s_\ell+1.5R_\delta}{\Delta} + \sqrt{k} \right)^{k-1} \quad&\textrm{if } s_\ell \ge  \Delta\sqrt{k}+1.5R_\delta
	\end{cases} \,,
	\]
	where $\VolK$ is the volume of the $k$-dimensional unit ball.
\end{lemma}
\begin{proof}
	This is an essentially standard packing argument, made slightly more complicated (when $s_\ell$ is large) since we are considering intersections against a sphere rather than a ball. For radii $0\le r_1 \le r_2$, denote the (closed) annulus by 
	\[
	\m{A}(r_1,r_2) = \m{B}(\0,r_2)\setminus \mathrm{int}(\m{B}(\0,r_1)) \,.
	\]
	Observe that $\SphereK \cap \m{B}\left(\frac{\Delta}{s_\ell}\bt,\frac{1.5R_\delta}{s_\ell}\right)\ne \emptyset$ implies\footnote{$[x]_+$ denotes the positive part of $x$, namely, $[x]_+=\max\{x,0\}$.} $\frac{\Delta}{s_\ell}\bt \in \m{A}\left(\left[1-\frac{1.5R_\delta}{s_\ell}\right]_{+},1+\frac{1.5R_\delta}{s_\ell}\right)$, so 
	\[
	N_\ell \le \left| \ZZ^k  \cap \m{A}\left( \left[ \frac{s_\ell - 1.5R_\delta }{\Delta} \right]_{+}, \frac{s_\ell + 1.5R_\delta }{\Delta} \right)\right|\,.
	\]
	Next, we use the following packing argument: the sets $\ZZ^k + \left(-\frac12,\frac12 \right)^k$ are disjoint, so that if $\bt\in \m{A}(r_1,r_2)$ then $\bt+\left(-\frac12,\frac12\right)^k \subset \m{A}\left(\left[r_1-\sqrt{k}\right]_+,r_2+\sqrt{k}\right)$. Therefore, by a volume comparison,
	\[
	\left| \ZZ^k \cap \m{A}(r_1,r_2)\right| \le \frac{\Vol_k\left(\m{A}([r_1-\sqrt{k}]_+,r_2+\sqrt{k})\right)}{\Vol_{k}\left(\left(-\frac12,\frac12 \right)^k\right)} = \VolK \cdot \left( (r_2+\sqrt{k})^k - [r_1-\sqrt{k}]_+^k \right) \,.
	\]
	Set $r_2=\frac{s_\ell+1.5R_\delta}{\Delta}$ and $r_1=\left[ \frac{s_\ell - 1.5R_\delta }{\Delta} \right]_{+}$, so that 
	\[
	N_\ell \le \VolK \cdot \left( \left( \frac{s_\ell+1.5R_\delta}{\Delta} + \sqrt{k} \right)^k - \left(\left[ \frac{s_\ell-1.5R_\delta}{\Delta}-\sqrt{k}\right]_+\right)^k \right) \,.
	\]
	The second term is non-zero if and only if $s_\ell\ge  \Delta\sqrt{k}+1.5R_\delta$; the claimed bound follows from the inequality $|a^k-b^k| \le k|b-a|\max\{|a|,|b|\}^{k-1}$.
\end{proof}

We now conclude the proof of Lemma~\ref{lem:good-whp-helper}. Recall, by Eq. (\ref{eq:good-helper-bound1}),
that our goal is to bound $\sum_{\ell=1}^T p_\ell$, where $p_\ell$ is bounded in Eq. (\ref{eq:helper-pell-bound}). We treat separately small and large terms in the sum. 
\begin{itemize}
	\item {\bf Small terms:} $\ell$-s such that $s_\ell < \Delta \sqrt{k} + 1.5R_\delta$. Note that there are $\le \frac{\Delta\sqrt{k} + 1.5R_\delta}{R_\delta} +1  \lesssim 1 + \frac{\Delta\sqrt{k}}{R_\delta}$ such terms.
	Bound
	\[
	\VolK \le (C/\sqrt{k})^{k},\quad \frac{1.5R_\delta}{s_\ell} \le \frac{1.5R_\delta}{A} \lesssim \frac{R_\delta\sqrt{\log k}}{\Delta_\delta\sqrt{k} }
	\]
	(recall {$s_\ell\ge A$} and the definition of $A$ in Eq. (\ref{eq:Line-tilde})). Assuming 
	\begin{equation}
		\label{eq:Cond2}
		\Delta \ge \frac{R_\delta}{\sqrt{k}}\,,
	\end{equation}
	we have
	\[
	\frac{s_\ell + 1.5R_\delta }{\Delta} + \sqrt{k} \le \frac{(\Delta \sqrt{k} + 1.5R_\delta) + 1.5R_\delta }{\Delta} + \sqrt{k} \lesssim \sqrt{k} + \frac{R_\delta}{\Delta} \overset{(\ref{eq:Cond2})}{\lesssim} \sqrt{k}\,.
	\]
	Plugging into Eq. (\ref{eq:helper-pell-bound}) and Lemma~\ref{lem:packing-bound},
	\[
	p_\ell \lesssim C^k (1/\sqrt{k})^k \cdot \left( \sqrt{k}\right)^k \left( \frac{R_\delta\sqrt{\log k}}{\Delta_\delta\sqrt{k} } \right)^{k-1} 
	\lesssim 
	 \left( C\frac{R_\delta\sqrt{\log k}}{\Delta_\delta\sqrt{k} } \right)^{k-1}\,,
	\]
	for some universal $C$. 
	Recalling again that there are $\lesssim 1 + \frac{\Delta\sqrt{k}}{R_\delta} \overset{(\ref{eq:Cond2})}{\lesssim} \frac{\Delta\sqrt{k}}{R_\delta}$ such terms, and that, by definition (Eq. (\ref{eq:def-Delta-R-delta})),
	\begin{equation}
		\label{eq:def-Delta-delta}
			\Delta = 2\Delta_\delta + 2z_{\infty}(\delta/2) \lesssim \Delta_\delta + \sqrt{ \log k \vee \log(1/\delta)}\,,
	\end{equation}
	the total sum of the small terms is
	\begin{align*}
		&\lesssim \sqrt{\log k}\left( C\frac{R_\delta\sqrt{\log k}}{\Delta_\delta\sqrt{k} } \right)^{k-2} + \sqrt{ \log k \vee \log(1/\delta)} \cdot \frac{\sqrt{k}}{R_\delta }\cdot \left( C\frac{R_\delta\sqrt{\log k}}{\Delta_\delta\sqrt{k} } \right)^{k-1}\\
		&\lesssim \sqrt{\log k}\left( C\frac{R_\delta\sqrt{\log k}}{\Delta_\delta\sqrt{k} } \right)^{k-2} + \sqrt{ k\log k}  \left( C\frac{R_\delta\sqrt{\log k}}{\Delta_\delta\sqrt{k} } \right)^{k-1}\,,
	\end{align*}
	where, for the second inequality, we used $R_\delta\ge z_2(\delta/3)\gtrsim\sqrt{\log(1/\delta)}$.
	Consequently, whenever
	\begin{equation*}
		\Delta_\delta \gtrsim \frac{R_\delta\sqrt{\log k}}{\sqrt{k}} \,,
	\end{equation*}
	the sum is exponentially decaying in $k$, and in particular $O(k^{-10})$. Again, recalling Eq. (\ref{eq:def-Delta-delta}), the following condition on $\Delta$ is sufficient to get exponential decay:
	\begin{equation}
		\label{eq:Cond3}
		\Delta \gtrsim \frac{R_\delta\sqrt{\log k}}{\sqrt{k}} + \sqrt{\log k \vee \log(1/\delta)} \,.
	\end{equation}

	
	\item {\bf Large terms:} such that $s_\ell \ge \Delta \sqrt{k} + 1.5R_\delta$. Note that there are $\lesssim B/R_\delta \lesssim \frac{\sqrt{\Spike \log (1/\delta)}}{R_\delta}$ such terms (recall the definition of $B$ in Eq. (\ref{eq:Line-tilde})). Bounding $\frac{s_\ell + 1.5R_\delta }{\Delta} + \sqrt{k} \le \frac{2s_{\ell}}{\Delta}$, we estimate, using $\VolK\le (C/\sqrt{k})^k$ and assuming condition (\ref{eq:Cond2}),
	\[
	N_\ell \le k\cdot \VolK \cdot \left( \frac{3R_\delta}{\Delta} + 2\sqrt{k} \right)\left( \frac{s_\ell+1.5R_\delta}{\Delta} + \sqrt{k} \right)^{k-1} 
	\lesssim k (C/\sqrt{k})^k \sqrt{k} \left( \frac{2s_\ell}{\Delta}\right)^{k-1}\,,
	\]
	so that, using Eq. (\ref{eq:helper-pell-bound}),
	\[
	p_\ell \le N_\ell \left(\frac{1.5R_\delta}{s_\ell}\right)^{k-1} 
	\lesssim k \left( 3C \frac{R_\delta}{\Delta\sqrt{k}} \right)^{k-1} \,.
	\]
	Again, since there 
	are $\lesssim \frac{\sqrt{\Spike \log(1/\delta)}}{R_\delta}$ such terms, the total contribution is
	\[
	\lesssim \frac{k}{R_\delta} \sqrt{\Spike \log(1/\delta)} \cdot \left( 3C \frac{R_\delta}{\Delta\sqrt{k}} \right)^{k-1} \le
	k 
\left( 3C \frac{R_\delta \Spike^{\frac{1}{2(k-1)}}}{\Delta\sqrt{k}} \right)^{k-1} \,,
	\]
	where, for the second inequality, we again used $R_\delta\gtrsim\sqrt{\log(1/\delta)}$.
	This is exponentially decreasing in $k$ whenever
	\begin{equation}
		\label{eq:Cond4}
		\Delta \gtrsim \frac{R_\delta}{\sqrt{k}} \Spike^{\frac{1}{2(k-1)}} \,.
	\end{equation}
\end{itemize}

We finish by simplifying conditions (\ref{eq:Cond3}) and (\ref{eq:Cond4}) further. Setting $\delta=k^{-M^2}$, we may estimate
\[
R_{\delta} \approx \sqrt{k}\vee \sqrt{\log(1/\delta)} = \sqrt{k}\vee (M\sqrt{\log k}),
\]
so that for large $k\ge k_0(M)$, $R_\delta \approx \sqrt{k}$. Thus, (\ref{eq:Cond3}) reads $\Delta \gtrsim M\sqrt{\log k}$, and (\ref{eq:Cond4}) reads $\Delta \gtrsim \Spike^{\frac{1}{2(k-1)}}$.

\section{Proof of Lemma~\ref{lem:spectral-gap} }
\label{sec:proof-lem:spectral-gap}

Recall the choice of $R$ from Eq. (\ref{eq:R-def}). Decompose, rather arbitrarily, $R^2 = k + B$, so that $3k\le B \le C_3 k$ for some $C_3$. Note that if $g_1,\ldots,g_k\sim \m{N}(0,1)$, then $\Pr\left( \sum_{i=2}^k g_i^2 \le B \right) \ge 1-e^{-C_4k}$ for some $C_4>0$. Following Proposition~\ref{obs:trunc-eigs} Eq. (\ref{eq:trunc-eigs}),
\[
\lambda_1(\bSigmaTruc) = \Expt \left[ \lambda_i g_i^2 \,\Big| \, \sum_{i=1}^k \lambda_i g_i^2 \le R^2\right] = \frac{\Expt\left[ (1+\Spike)g_1^2 \cdot \Ind_{(1+\Spike)g_1^2 + \sum_{i=2}^k g_i^2 \le k + B} \right]}{\Pr\left( (1+\Spike)g_1^2 + \sum_{i=2}^k g_i^2 \le k + B  \right)} \,.
\]
Clearly, $\Ind_{(1+\Spike)g_1^2+\sum_{i=2}^k g_i^2 \le k + B} \ge \Ind_{(1+\Spike)g_1^2 \le k}\cdot \Ind_{\sum_{i=2}^k g_i^2 \le B}$, therefore,
\begin{align*}
	\Expt\left[ (1+\Spike)g_1^2 \cdot \Ind_{(1+\Spike)g_1^2 + \sum_{i=2}^k g_i^2 \le k + B} \right] 
	&\ge \Expt\left[ (1+\Spike)g_1^2 \cdot \Ind_{(1+\Spike)g_1^2 \le k}\cdot \Ind_{\sum_{i=2}^k g_i^2 \le B} \right]\\
	&\overset{(\star)}{=} \Expt\left[ (1+\Spike)g_1^2 \cdot \Ind_{(1+\Spike)g_1^2 \le k}\right] \cdot \Expt\left[\Ind_{\sum_{i=2}^k g_i^2 \le B} \right]\\
	&\ge \Expt\left[ (1+\Spike)g_1^2 \cdot \Ind_{(1+\Spike)g_1^2 \le k}\right] (1-e^{-C_4 k})\,,
\end{align*}
where $(\star)$ holds since this is the product of independent random variables. Furthermore, clearly, 
\[
\Pr\left( (1+\Spike)g_1^2 + \sum_{i=1}^k g_i^2 \le k + B  \right)\le \Pr\left( (1+\Spike)g_1^2 \le k + B  \right) \le \Pr\left( (1+\Spike)g_1^2 \le (1+C_3)k  \right)\,.
\]
Let $g_1^2 =: W\sim \chi^2(1)$, so that, finally,
\begin{equation}
	\lambda_1(\bSigmaTruc)  \ge (1-e^{-C_4 k})\frac{\Expt\left[ (1+\Spike)W \cdot \Ind_{W \le \frac{k}{1+\Spike}}\right]}{\Pr\left( W \le \frac{(1+C_3)k}{1+\Spike}  \right)} \,.
\end{equation}
We continue case-by-case, depending on the magnitude of $\Spike$:

\begin{enumerate}[label=(\roman*)]
	\item Suppose that $e^{-C_1 k} \le \Spike \le 1$, where $C_1$ is a sufficiently small constant. Since $W$ has an exponential tail and $\Expt[W]=1$, there is some $C_5$ such that 
	\[
	\frac{\Expt\left[ W \cdot \Ind_{W \le \frac{k}{1+\Spike}}\right]}{\Pr\left( W \le \frac{(1+C_3)k}{1+\Spike}  \right)} \ge \frac{\Expt\left[ W \cdot \Ind_{W \le \frac{k}{2}}\right]}{\Pr\left( W \le {(1+C_3)k} \right)} \ge 1-e^{-C_5 k} \,,
	\]
	therefore
	\[
	\lambda_1(\bSigmaTruc) \ge (1+\Spike)(1-e^{-C_4 k})(1-e^{-C_5 k}) \ge 1+C\Spike
	\]
	for small enough $C$, whenever $C_1$ is chosen sufficiently small compared to $C_4,C_5$.
	\label{item:gap-step1}
	\item Note that by \cite{palombi2012numerical}, $\lambda_1(\bSigmaTruc)$ increases with $\Spike$. Consequently, for all $\Spike\ge 1$, \ref{item:gap-step1} implies that $\lambda_1(\bSigmaTruc) \ge 1+C$.  
	Now, suppose that $1 \le \Spike \le Ak - 1$, where $A$ is such that for all $A'\le A$,
	\[
	\frac{\Expt\left[ W \cdot \Ind_{W \le \frac{1}{A'}}\right]}{\Pr\left( W \le \frac{(1+C_3)}{A'}  \right)} \ge \frac{4}{5}\,.
	\]
	Note that such $A$ indeed exists, since the above ratio $\to 1$ as $A'\to 0$. Then
	\[
	\lambda_1(\bSigmaTruc) 
	\ge (1+\Spike)(1-e^{-C_4 k}) \frac{\Expt\left[ W \cdot \Ind_{W \le \frac{k}{1+\Spike}}\right]}{\Pr\left( W \le \frac{(1+C_3)k}{1+\Spike}  \right)} 
	\ge \frac45(1-e^{-C_4})(1+\Spike) \ge C'(1+\Spike) \,,
	\]
	so that $\lambda_1(\bSigmaTruc)\ge (1+C)\vee C'(1+\Spike)$. Consequently, $\lambda_1(\bSigmaTruc)\ge 1 + C''\Spike$ for some other $C''$.
	\item $\Spike \ge Ak-1$. Consider the function
	\[
	F(h) = \frac{1}{h} \cdot  \frac{\Expt\left[ W\cdot \Ind_{W\le h}\right]}{\Pr\left( W\le (1+C_3)h\right)} \,,
	\]
	{so that 
	\[
	\lambda_1(\bSigmaTruc) \ge (1-e^{-C_4 k})\cdot k\cdot F\left(\frac{k}{1+\Spike}\right) \ge k \cdot (1-e^{-C_4}) \cdot \inf_{h\le 1/A} F(h) \,.
	\]
}
	We are done if we show that the infimum is non-zero, and it clearly suffices to show that $\lim_{h\to 0+}F(h)>0$. To do this, recall that $W\sim \chi^2(1)$ has a density $f_W(w)\propto w^{-1/2}e^{-w/2}$ supported on $w\ge 0$. Therefore, as $h \to 0+$,
	\[
	\Expt\left[ W\cdot \Ind_{W\le h}\right] \sim Ch^{3/2},\quad \Pr\left( W\le (1+C_3)h\right) \sim C'h^{1/2}\,,
	\]
	so that indeed $F(h)\sim C''$ as $h\to 0+$.
\end{enumerate}

\section{Proof of Theorem~\ref{thm:main}}\label{sec:proof-thm:main}

We shall use the following eigenvalue perturbation result \cite[Theorem 8.5]{wainwright2019high}:
\begin{lemma}
	\label{lem:wainwright-perturbation}
	Let $\bA$ be positive semidefinite, with a positive spectral gap: $\delta := \lambda_1(\bA)-\lambda_2(\bA)>0$. Let $\SpikeVec\in\SphereK$ be its largest eigenvector. 
	Suppose that $\widehat{\bA}$ is positive semidefinite with $\|\bA-\widehat{\bA}\|\le \delta/4$. Let $\SpikeVecEst$ be its maximal eigenvector, with the sign chosen such that $\langle \SpikeVec,\SpikeVecEst\rangle \ge 0$ (part of the claim is that the largest eigenspace of $\widehat{\bA}$ is 1-dimensional). Then
	\[
	\|\SpikeVec-\SpikeVecEst\| \le \frac{4}{\delta} \cdot \left\|\left(\bI-\SpikeVec\SpikeVec^\T\right) (\bA-\widehat{\bA}) \SpikeVec\right\| \,,
	\]
	where $\bI-\SpikeVec\SpikeVec^\T$ is the projection onto the orthogonal complement of $\SpikeVec$.
\end{lemma}

We apply Lemma~\ref{lem:wainwright-perturbation} with $\bA=\bSigmaTruc$ and $\widehat{\bA}=\hbSigma$, using the error bounds developed so far.
For brevity, denote
\begin{equation}
	\delta := \lambda_1(\bSigmaTruc) - \lambda_2(\bSigmaTruc) \,.
\end{equation}

\begin{lemma} \label{lem:thm-proof-helper}
		Assume the setup of Lemma~\ref{lem:good-whp}, with $\Delta$ large, $M> \sqrt{12}$, and $\SpikeVec\in\m{U}_M$. 
		
		Assume either of the following conditions hold:
		\begin{itemize}
			\item $1\le \Spike \le k$ and $n\gtrsim \frac{k}{\sqrt{\Spike}}\vee \log k$.
			\item $k \le \Spike \lesssim k^{2M^2-21}$ and $n\gtrsim \sqrt{\Spike}$. 
		\end{itemize}
	Then with probability $1-O(k^{-10})$, one has $|\bSigmaTruc-\hbSigma\| \le \delta/4$
\end{lemma}
\begin{proof}
	We consider two cases:
	\begin{itemize}
		\item Suppose that $1\le \Spike \le k$. By Lemma~\ref{lem:Pball}, $\PBall\approx 1$, and by Lemma~\ref{lem:spectral-gap}, $\delta\gtrsim  \Spike$. By Lemmas \ref{lem:epsEst-bound} and \ref{lem:epsPick-bound}, it holds with probability $1-O(k^{-10})$ that
		\begin{align*}
			\|\bSigmaTruc-\hbSigma\|
			\lesssim  \Spike \sqrt{\frac{\log k}{n}} + \sqrt{\Spike} \left( \frac{k}{n}\vee \sqrt{\frac{k}{n}} \right) +  k^{-M^2+12}\,.
		\end{align*}
	Consequently, when $n\gtrsim \log k$, $n\gtrsim \frac{k}{\sqrt{\Spike}}$ and $n\gtrsim \frac{k}{\Spike}$ (the last condition is redundant, since $\Spike\ge 1$), it holds that, for large enough $k$, $\|\bSigmaTruc-\hbSigma\|\le \delta/4$.
		
		\item Suppose that $k\le \Spike \le k^{2M^2-21}$. By Lemma~\ref{lem:Pball}, {$\PBall\approx \sqrt{k/\Spike}$}, and by Lemma~\ref{lem:spectral-gap}, $\delta\gtrsim  k$. By Lemmas \ref{lem:epsEst-bound} and \ref{lem:epsPick-bound}, it holds with probability $1-O(k^{-10})$, provided that $\PBall n \gtrsim \log k \implies n\gtrsim \sqrt{\frac{\Spike}{k}}\log k$, that
	 \begin{align*}
	 	\|\bSigmaTruc-\hbSigma\|
	 	\lesssim  k\sqrt{ \frac{\sqrt{\Spike}\log k}{n\sqrt{k}} } + k\left( \frac{\sqrt{\Spike}}{n} \vee \sqrt{ \frac{\sqrt{\Spike}}{n\sqrt{k}} } \right)
	 	+ k \cdot \Spike^{1/2}k^{-M^2+10.5} \,.
	 \end{align*}
 	Consequently, whenever $n\gtrsim \sqrt{\frac{\Spike}{k}}\log k$, $n\gtrsim \sqrt{\Spike}$ (the first condition is redundant), and $\Spike \lesssim k^{-2M^2 + 21}$, it holds that $\|\bSigmaTruc-\hbSigma\|\le \delta/4$. 
	\end{itemize}
\end{proof}

\paragraph{Proof of Theorem~\ref{thm:main}.} We apply Lemma~\ref{lem:wainwright-perturbation}. Write, as before, $\hbSigma-\bSigmaTruc=(\hbSigma-\hbSigmaX)+(\hbSigmaX-\bSigmaTruc)$, so
	$\left\|\left(\bI-\SpikeVec\SpikeVec^\T\right) (\hbSigma-\bSigmaTruc) \SpikeVec\right\| \le \epsPick + \left\|\left(\bI-\SpikeVec\SpikeVec^\T\right) (\hbSigmaX-\bSigmaTruc) \SpikeVec\right\|$. Using the decomposition Eq. (\ref{eq:hbSigmaX-decomposition}), and recalling the notation in Eq.~(\ref{eq:hbSigmaX-eps}), we conclude that under the conditions of Lemma~\ref{lem:thm-proof-helper}, with probability $1-O(k^{-10})$, 
	\[
	\|\SpikeVec-\SpikeVecEst\|\le \frac{4}{\delta}\left\|\left(\bI-\SpikeVec\SpikeVec^\T\right) (\hbSigma-\bSigmaTruc) \SpikeVec\right\| \le \frac{4}{\delta} \left( \epsPick + 2\eps_2 + \eps_3 \right)\,.
	\]
	Note that the term $\eps_1$ does not appear, since it corresponds to a components of the difference $\hbSigmaX-\bSigmaTruc$ which is parallel to $\SpikeVec$. Using Lemmas \ref{lem:epsPick-bound}, \ref{lem:eps2} and \ref{lem:eps3}, $\delta\gtrsim k\vee \Spike$ (Lemma~\ref{lem:spectral-gap}), and $\PBall \approx 1\wedge \sqrt{k/\Spike}$, we conclude that the following holds with probability $1-O(k^{-10})$:
	\begin{itemize}
		\item Suppose that $1\le \Spike \le k$ and $n\gtrsim \frac{k}{\sqrt{\Spike}}\vee \log k$. Then
		\[
		\|\SpikeVec-\SpikeVecEst\| \lesssim \frac{1}{\sqrt{\Spike}} \left( \frac{k}{n} \vee \sqrt{\frac{k}{n} } \right) + \frac{1}{\Spike} \cdot k^{-M^2+12} \,.
		\]
		Note that the requirement $n\gtrsim \frac{k}{\sqrt{\Spike}}$ may effectively be omitted from the statement of the Theorem. The reason is that a bound of the form $\|\SpikeVec-\SpikeVecEst\|\le B$, for any $B\ge 2$, is completely vacuous (since $\SpikeVec,\SpikeVecEst$ are unit vectors). As we are not keeping track of the exact constants, it suffices to note that the first term in the upper bound becomes meaningful only when $n\gtrsim \frac{k}{\sqrt{\Spike}}$.
		\item Suppose that $k\le \Spike \lesssim k^{2M^2-21}$ and $n\gtrsim \sqrt{\Spike}$. Then 
		\[
		\|\SpikeVec-\SpikeVecEst\| \lesssim \frac{\sqrt{\Spike}}{n} + \sqrt{ \frac{1}{n}\sqrt{\frac{\Spike}{k}} } + \Spike^{1/2}k^{-M^2+10.5} \,.
		\]
		The  requirement $n\gtrsim \sqrt{\Spike}$ is omitted from the statement of the  Theorem, for the same reason as in the previous case.
	\end{itemize}

\section{Proof of additional lemmas}

In this section, we provide several short proofs, that were omitted from  the main text due to space constraints. 

\subsection{Proof of Lemma~\ref{lem:Pball}}
\label{sec:proof-lem:Pball}

	Upper bound: ${\langle \bu,\bX\rangle \sim \m{N}(0,1+\Spike)}$; clearly, $\bX\in \m{B}(\0,R)$ implies $|\langle \bu,\bX\rangle|\le R$, so ${\PBall\le \Pr(|\langle \bu,\bX\rangle|\le R)=\erf(R/\sqrt{2(1+\Spike)})}$.
	
Lower bound: Writing $\bX=\sqrt{\Spike}\xi\SpikeVec + \bZ$, 
the event $\{\sqrt{\Spike}|\xi|\le 2\sqrt{k}\}\cap \{\|\bZ\|\le z_2(0.1)\}$ implies $\bX\in\m{B}(\0,R)$. Thus,
\begin{align*}
	\PBall
	&\ge \Pr\left( \{\|\bZ\|\le z_2(0.1)\}\cap \{\sqrt{\Spike}|\xi|\le 2\sqrt{k}\}\right) \\
	&\overset{(\star)}{=} \Pr\left(\|\bZ\|\le z_2(0.1)\right) \cdot \Pr\left(\sqrt{\Spike}|\xi|\le 2\sqrt{k}\right) \\
	&\ge 0.9\cdot \erf\left( \sqrt{\frac{2k}{\Spike}}\right)\,,
\end{align*}
where $(\star)$ follows since these are independent events.

\subsection{Proof of Lemma~\ref{lem:simple-convex}}
\label{sec:proof-lem:simple-convex}

	By the Gaussian correlation inequality \cite{royen2014,latala2017royen}, for ${f,h:\RR^k\to \RR}$  quasiconcave\footnote{$h$ is quasiconcave if $h(tx+(1-t)y)\ge \min\{h(x),h(y)\}$ for all $t\in[0,1]$. Note that: (i) A concave function is quasiconcave; (ii) The indicator function of a convex set is quasiconcave (but {\it not} concave).}, and \emph{one} of whom symmetric, ${\Expt\left[ f(\bX)h(\bX) \right] \ge \Expt\left[f(\bX)\right]\Expt\left[h(\bX)\right]}$. Consequently, for $g$ convex,
${\Expt\left[ g(\bX)h(\bX) \right] \le \Expt\left[ g(\bX) \right]\Expt \left[h(\bX) \right]}$. 
The Lemma follows by 
taking $h(\bX)=\Ind_{\bX\in\m{B}(\0,R)}$.

\subsection{Proof of Lemma~\ref{lem:epsPick-bound}}
\label{sec:proof-lem:epsPick-bound}

	First, suppose that $n\le k^{M^2-10}$. Then by Lemma~\ref{lem:good-whp} and Markov's inequality, with probability $1-O(nk^{-M^2})=1-O(k^{-10})$, it holds that $|\Kbad|=0$, and consequently (Eq. (\ref{eq:epsPick})), $\epsPick=0$. Next, suppose that $n\ge k^{M^2-10}$. By Chernoff's inequality (Lemma~\ref{lem:chernoff}), Lemma~\ref{lem:good-whp} and Lemma~\ref{lem:Nball} it hold with probability $1-O(k^{-10})$ that
{
\begin{align*}
	\epsPick 
	&\lesssim \frac{k}{n\PBall} |\Kbad| 
	\lesssim \frac{k}{n\PBall} \left( nk^{-M^2} \vee \log k \right) 
	\le \frac{k^{-M^2+1}}{\PBall} + \frac{k\log k}{n\PBall} 
	\le
	\frac{2k^{-M^2+12}}{\PBall} \,.
\end{align*}
}
\section{Auxiliary technical lemmas}
\label{sec:aux}

The following are tail bounds for some norms of a Gaussian random vector:
\begin{lemma}
	\label{lem:gaussian-tail-bounds}
	Let $\bZ\sim \m{N}(\0,\bI_k)$. Then
	\begin{enumerate}[label=(\roman*)]
		\item $\ell_2$:
		\begin{equation*}
			\begin{split}
				\Pr\left( \|\bZ\|_2^2 \ge k + 2\sqrt{kx} + 2x \right) \le e^{-x} \,, \quad
				\Pr\left( \|\bZ\|_2^2 \le k - 2\sqrt{kx}  \right) \le e^{-x} \,.
			\end{split}
		\end{equation*}
	\label{item:z2}
		\item $\ell_\infty$:
		\[
		\Pr\left( \|\bZ\|_{\infty} \ge \sqrt{2\log k} + x \right) \le 2e^{-\frac12 x^2} \,.
		\]
		\label{item:zInf}
	\end{enumerate}
\end{lemma}
\begin{proof}
	\ref{item:z2} is the well-known inequality of Laurent and Massart \cite[Lemma 1]{laurent2000adaptive}.
	\ref{item:zInf} is a special case of the Borell-TIS inequality; alternatively, it follows from the Gaussian Lipschitz concentration inequality, e.g. \cite[Theorem 5.6]{boucheron2013concentration}, with the elementary bound ${\Expt[\max_{1\le i \le k}Z_i] \le \sqrt{2\log k}}$.
\end{proof}

The following is an immediate corollary:
\begin{lemma}[$\ell_\infty$ bound for a uniform vector in $\SphereK$]
	\label{lem:unif-inf}
	Suppose that $\bu\sim \Unif(\SphereK)$. There are absolute constants $C,c$ such that 
	\[
	\Pr\left( \|\bu\|_\infty \ge c\sqrt{\frac{\log k}{k}} \right)   \le Ck^{-10} \,.
	\]
\end{lemma}
\begin{proof}
	Let $\bZ\sim \m{N}(0,\bI_k)$, so that $\bu \overset{d}{=} \bZ/\|\bZ\|$. {Choosing $c$ large enough and using Lemma~\ref{lem:gaussian-tail-bounds},
	\[
	\Pr\left( \|\bu\|_\infty \ge c\sqrt{\frac{\log k}{k}} \right) \le \Pr\left(\|\bZ\|\le \frac12 \sqrt{k}\right) + \Pr\left( \|\bZ\|_\infty \ge \frac12 c\sqrt{\log k}\right) \le e^{-\Omega(k)} + O(k^{-10})= O(k^{-10}) \,.
	\]
}
\end{proof}

Next is Chernoff's inequality for Bernoulli random variables \cite[Theorem 2.3.1, Exercise 2.3.2]{vershynin2018high}:
\begin{lemma}[Chernoff's inequality]\label{lem:chernoff}
	Let $X_1,\ldots,X_n$ be independent Bernoulli random variables. Set ${S_n=\sum_{i=1}^n X_i}$ and ${\mu=\Expt\left[S_n\right]}$. Then for all $\alpha>1$,
	\begin{equation*}
		\begin{split}
			\Pr\left( S_n\ge \alpha \mu \right) \le e^{-\mu} (\alpha/e)^{-\alpha\mu},\quad \Pr\left( S_n\le \mu/\alpha \right) \le  (e\alpha)^{\mu/\alpha} e^{-\mu } \,.
		\end{split}
	\end{equation*}
	In particular, there are some absolute constants $c,C$ such that
	\[
	\Pr\left( S_n\ge C \mu \right)\le e^{-c\mu},\quad \Pr\left( S_n\le \mu/C \right) \le e^{-c\mu} \,.
	\]
\end{lemma}

We recall some properties of the sub-Gaussian and sub-exponential norms. 
The following is taken from \cite[Chapter 2]{vershynin2018high}:

\begin{definition}[Orlicz norm]\label{def:orlicz-norm}
	Let $\psi:[0,\infty)\to[0,\infty)$ be convex, strictly increasing such that
	\[
	\psi(0)=0,\quad \psi(x)\to\infty \,\textrm{ as }\,x\to \infty \,.
	\]
	For a random variable $X$, define its {\it $\psi$-Orlicz norm} by
	\begin{equation}
		\label{eq:orlicz-def}
		\|X\|_{\psi} = \inf \left\{ t\,:\,\Expt \psi\left( |X|/t \right) \le 1 \right\}\,.
	\end{equation}
	For a random vector $\bX$, its $\psi$-Orlicz norm is $\|\bX\|_{\psi}=\sup_{\bv\in\SphereK}\|\langle \bX,\bv\rangle \|_{\psi}$.
\end{definition}

It is not hard to show that $\|\cdot \|_{\psi}$ is indeed a norm. The choices 
\[
\psi_2(x) = e^{x^2}-1,\quad \psi_1(x)=e^x-1
\]
correspond to the sub-Gaussian and sub-exponential norms respectively. $X$ is sub-Gaussian (resp. sub-exponential) in the ``usual'' sense if and only if $\|X\|_{\psi_2}<\infty$ (resp. $\|X\|_{\psi_1}<\infty$); see \cite[Chapter 2]{vershynin2018high} for more background. We briefly mention some properties of these norms that are used in the paper:
\begin{lemma}\label{lem:orlicz-props}
	The following holds:
	\begin{enumerate}
		\item \label{item:orlicz1} $\|X^2\|_{\psi_1}=\|X\|_{\psi_2}^2$.
		\item \label{item:orlicz2} $\|XY\|_{\psi_1}\le \|X\|_{\psi_2}\|Y\|_{\psi_2}$ ($X,Y$ {\bf do not} need to be independent).
		\item \label{item:orlicz3} {\it Centralization lemma}: $\|X-\Expt[X]\|_{\psi_i} \le \|X\|_{\psi_i}$ for $i=1,2$. 
		\item \label{item:orlicz4} For independent $X_1,\ldots,X_n$: $\left\| \sum_{i=1}^n X_i \right\|_{\psi_2}^2 \le C\sum_{i=1}^n \|X_i\|_{\psi_2}^2$, for some $C>0$ universal.
		\item \label{item:orlicz5} {\it Hoeffding's lemma}: for a bounded random variable, $\|\bX\|_{\psi_2}\le C\|\bX\|_{\infty}$. 
	\end{enumerate}
\end{lemma}

The following is Bernstein's inequality for sums of independent sub-exponential random variables \cite[Theorem 2.8.1]{vershynin2018high}:
\begin{lemma}[Bernstein's inequality]\label{lem:bernstein}
	Let $X_1,\ldots,X_n$ be independent and sub-exponential. Set ${S_n=\sum_{i=1}^n X_i}$. Then for all $t\ge 0$,
	\[
	\Pr\left( \left| S_n-\Expt[S_n] \right| \ge t \right) \le 2\exp\left[ -c\min\left( \frac{t^2}{\sum_{i=1}^n \|X_i\|_{\psi_1}^2}, \frac{t}{\max_{1\le i\le n}\|X_i\|_{\psi_1}} \right) \right]\,,
	\]
	where $c>0$ is an absolute constant.
\end{lemma}

Lastly, we cite a concentration inequality for sample covariance matrices with sub-Gaussian measurements \cite[Theorem 4.6.1]{vershynin2018high}:
\begin{lemma}\label{lem:vershynin}
	Let $\bx_1,\ldots,\bx_n\in \RR^k$ be independent, centered and sub-Gaussian. Denote $K=\max_{1\le i\le n}\|\bx_i\|_{\psi_2}$.
	
	Let $\hbSigma=\frac1n \sum_{i=1}^n\bx_i\bx_i^\T$ be the sample covariance. There is $C>0$ such that with probability at least $1-2e^{-t^2}$, 
	\[
	\left\| \hbSigma - \Expt\left[\hbSigma\right] \right\| \le K^2\max\{\delta,\delta^2\},\quad \textrm{where } \delta=C\left( \sqrt{\frac{k}{n}} + \frac{t}{\sqrt{n}} \right)\,.
	\]
	(Note that in \cite[Theorem 4.6.1]{vershynin2018high}, the result is stated for isotropic vectors, meaning $\Cov(\bx_i)=\bI$. However, the proof goes through, verbatim, also without this assumption.)
\end{lemma}

%
%
\end{document}